%% file: main.tex
\documentclass[a4paper,article,10pt]{memoir}
\usepackage[english]{babel}
\usepackage{CCLAuthor}
\usepackage[reqno]{amsmath}
\usepackage{amssymb}
\usepackage{amsthm}  
\theoremstyle{plain}
\newtheorem{theorem}{Theorem}

\newtheorem{lemma}[theorem]{Lemma}
\newtheorem{corollary}[theorem]{Corollary}
\theoremstyle{definition}
\newtheorem{definition}[theorem]{Definition}
\newtheorem{remark}[theorem]{Remark}

\numberwithin{theorem}{chapter}
\usepackage[round,authoryear,comma]{natbib} 
\usepackage[]{graphicx} 
\usepackage{xcolor}
\usepackage{tikz}
\usepackage{pgfplots}
\pgfplotsset{compat=1.18}

\usepackage{enumitem}
\graphicspath{{./Figures/}}

\newtheorem{exercise}[theorem]{Exercise}

\newcommand{\dd}{\mathop{}\!\mathrm{d}}
\renewcommand{\L}{\mathcal{L}}
\newcommand{\F}{\mathcal{F}}
\newcommand{\ee}{\mathrm{e}}
\renewcommand{\Re}{\mathop{}\!\mathrm{Re}}
\renewcommand{\Im}{\mathop{}\!\mathrm{Im}}
\newcommand{\sign}{\mathop{}\!\mathrm{sign}}
\newcommand{\res}{R}
\newcommand{\gr}{r}

\begin{document}
%
%
%
\title{Age-Structured Population Dynamics\thanks{This chapter is intended to become part of a volume edited by D. Breda, R. Vermiglio and J. Wu in the Springer Book Series CISM International Centre for Mechanical Sciences, following the school ``Delays and Structures in Dynamical Systems: Modeling, Analysis and Numerical Methods'', November 20--24, 2023.}}

%
%
\author{%
    Odo Diekmann\textsuperscript{a}
    and 
    Francesca Scarabel\textsuperscript{b,c}
    \\ \smallskip\small
    \textsuperscript{a} 
    Department of Mathematics, Utrecht University, Utrecht, The Netherlands                
    \\
    \textsuperscript{b} 
    School of Mathematics, University of Leeds, Leeds, United Kingdom
    \\
    \textsuperscript{c}
    Computational Dynamics Laboratory,
    Department of Mathematics, Computer Science and Physics,
    University of Udine, Italy
    }
    \maketitle
%
%
%
%
    \begin{abstract}
    This chapter reviews some aspects of the theory of age-structured models of populations with finite maximum age. We formulate both the renewal equation for the birth rate and the partial differential equation for the age density, and show their equivalence. Next, we define and discuss central concepts in population dynamics, like the basic reproduction number $R_0$, the Malthusian parameter $\gr$, and the stable age distribution. We briefly review the sun-star theory that turns the birth term into a bounded additive perturbation, thus allowing to develop stability and bifurcation theory along standard lines. Finally, we review the pseudospectral approximation of the infinite-dimensional age-structured models by means of a finite system of ordinary differential equations, which allows to perform numerical bifurcation analysis with existing software tools. Here, Nicholson’s blowfly equation serves as a worked example.
    \end{abstract}

\CCLsection{Motivation for Structured Models (and the Particularities of Age)}
\label{s:introduction}

The size of a population changes in the course of time. Indeed, the population size decreases due to the death of individuals, but increases when new individuals are born. In between birth and death, an individual develops. The basic premise of structured population models \citep{DeRoosPersson2013,MetzDiekmann1986} is that the behavior of an individual (concerning, e.g., the consumption of food, reproduction, chance to become a victim of a predator) depends on its state (called $i$-state to indicate that it refers to an individual) and that we need to model how the $i$-state changes in time. In \citep{DiekmannScarabelSize} we shall take the \emph{size} of an individual as its $i$-state and introduce as one of the model ingredients a description of how the growth rate of an individual depends on its size and on the prevailing food concentration (we shall call the food concentration the \emph{environmental condition}, since it describes the relevant properties of the world in which the individual lives). In this chapter, however, we shall work with \emph{age} as a phenomenological and statistical proxy for a potentially very high-dimensional `true' $i$-state, i.e., an $i$-state with a physiological interpretation. This is of course reminiscent of human demography, where survival and cumulative offspring profiles provide a statistical description of the health and the fertility of individuals (in \cite{MetzDiekmann1986}, the first section of the chapter on age dependence is titled `Age as a substitute for comprehension').

In both chapters we consider one-sex models. More precisely, we `count' only females, and reproduction refers to mothers who produce daughters. Implicitly, we assume a fixed sex ratio.

The definition of `age' entails two key characteristics:
\emph{(i)} it equals zero at birth, and 
\emph{(ii)} it increases in time at the fixed speed $1$. 
One of the main reasons for introducing age structure in a model is that it allows to incorporate maturation delay. When in biology we make a distinction between juveniles and adults, this indicates in particular that the first cannot reproduce while the latter can. A newborn individual is a juvenile, and it reaches adulthood only at a certain age, so after some time has elapsed.

The \emph{environmental conditions} (i.e., all those external variables that influence the individual rates) can be:
\begin{enumerate}[noitemsep,nolistsep]
\item[\emph{(i)}] constant in time,
\item[\emph{(ii)}] variable in time (but independent of the population),
\item[\emph{(iii)}] (partly) determined by feedback.
\end{enumerate}

In the first two cases, we deal with a linear model. In the third case, the model is nonlinear, since density dependence is incorporated: individuals interact indirectly since they are both \emph{sensitive} to the environmental condition and they exert an \emph{impact} on it. For instance, for a prey, the predator density is (part of) the environmental condition: sensitivity refers to the risk of falling victim, while impact refers to the fact that a victim contributes to the predator density in the future, since it helps the predator to survive and reproduce (of course the immediate short-term effect may be that the predator is satiated and stops hunting). Likewise, for a predator the prey density is (part of) the environmental condition: sensitivity refers to the need to eat, and impact refers to the fact that, by consumption, the availability of food is decreased.

\bigskip
In the next section, we consider age-structured population dynamics in a constant environment. We introduce the age-specific survival probability~$\F$ and the expected age-specific cumulative number of offspring~$\L$ as the two model ingredients, and formulate the renewal equation for the $p$-level birth rate~$b$ as a first mathematical formulation of the model. We provide a definition of `delay equation' and show how to interpret the renewal equation (RE) as a delay equation. 

After sketching in Section~\ref{s:RE} how to solve the RE constructively, in Section~\ref{s:dynsyst} we describe the dynamical system perspective of delay equations (meaning, in this linear autonomous situation, that we introduce a semigroup of bounded linear operators and its infinitesimal generator). 

In Section~\ref{s:PDE} we replace as modeling ingredients $\F$ by the age-specific death rate $\mu(a)$ and $\L$ by the age-specific fertility $\beta(a)$, and use these new ingredients to give a reformulation of the model in terms of a partial differential equation (PDE) for the age-density. We delineate how the semigroups corresponding to, respectively, the RE and the PDE formulation relate to each other.

Section~\ref{s:asymptotic} focuses on the asymptotic large time behavior of the dynamical system, which is described in terms of the unique real root $\gr$ (a.k.a.\ the Malthusian parameter) of the Euler--Lotka characteristic equation. We also define the basic reproduction number $R_0$, and we show that $\sign(R_0-1)=\sign(\gr)$.

In Section~\ref{s:widening} we anticipate sun-star calculus by motivating the embedding of $L^1$ in the space of measures, while observing that the concomitant loss of strong continuity is compensated by the gain of a second, weaker, topology. 

In Section~\ref{s:nicholson} we introduce our key nonlinear example: Nicholson's blowfly model. We perform a pen and paper stability and bifurcation analysis of the model. While doing so, we observe the need for: \emph{(i)} theory to justify some of our assertions, and \emph{(ii)} tools to perform a more encompassing numerical bifurcation analysis. 
We discuss the latter in Section~\ref{s:pseudospectral}, where we briefly describe the pseudospectral approximation method that allows to approximate a general RE with a system of ordinary differential equations (ODE). The numerical bifurcation analysis can then be performed with widely available software for ODE, as we show by way of Nicholson's blowfly equation as an example. 

The chapter ends with three appendices. 
In Appendix~\ref{s:ODE} we formulate necessary and sufficient conditions on the model ingredients for the reduction of the infinite-dimensional formulations to a finite system of ODE.
Appendix~\ref{s:sun-star} is devoted to a brief sketch of sun-star calculus, which yields the required theoretical underpinning of linearized stability results. 
In Appendix~\ref{s:NGO} we focus on the basic reproduction number and the next generation operator. 

Throughout our presentation, we sacrifice generality in order to gain readability. To compensate, we include pointers to the literature dealing with more general settings, such as for instance equations with infinite delay.

\CCLsection{Constant Environment}
\label{s:constant_env}
In this section, we consider an age-structured population living in a constant environment.
The model is fully determined by the basic modeling ingredient $\L(\dd a)$, describing the \emph{expected} number of daughters produced by an individual of age $a$ in an infinitesimal time interval of length $\dd a$.\footnote{Note that, mathematically, $\L$ so defined describes a \emph{measure} on $\mathbb{R}_+$.}
The word `expected' here is used to account, in a deterministic framework, for potential randomness and heterogeneity of individuals. 
The expected total number of offspring in an individual's lifetime is called the \emph{basic reproduction number}, and is denoted by $R_0$. So, here $R_0 = \L(\infty) = \int_0^\infty \L(\dd a)$.

The basic model variable is the $p$-level ($p$ for `population') birth rate $b(t)$ at time $t$, which satisfies the \emph{renewal equation} (RE) 
\begin{equation} \label{RE-L}
    b(t) = \int_0^\infty b(t-a) \L(\dd a).
\end{equation}
When, by assumption, the measure $\L$ has a density, we can alternatively write \eqref{RE-L} as 
\begin{equation} \label{RE}
b(t) = \int_0^\infty b(t-a) k(a) \dd a,
\end{equation}
where $k=\L'$, with $\L'$ denoting the derivative of $\L$. The kernel $k$ is often specified as 
\begin{equation}\label{k}
k(a) := \beta(a)\F(a), 
\end{equation}
where $\F(a)$ is the \emph{survival probability} until age $a$, assumed to be a non-increasing function of age, and 
\begin{equation*}
    \beta(a) := \frac{1}{\F(a)}\L'(a)
\end{equation*}
is the age-specific \emph{fertility rate}.

We also define the age-specific \emph{death rate} 
\begin{equation}\label{mu}
    \mu(a) := - \frac{\dd}{\dd a} \ln \F(a).
\end{equation}
It follows that 
$\F$ satisfies $\F' = -\mu \F$ with initial condition $\F(0)=1$, and, consequently, 
\begin{equation}\label{F}
    \F(a) = \mathrm{e}^{-\int_0^a \mu(\alpha)\dd \alpha}. 
\end{equation}

For the classical theory of RE we refer to \cite{BellmanCooke1963, Feller1941, Lotka1939}.
Equation \eqref{RE} is an example of a \emph{delay equation}, by which we mean a rule for extending a function of time towards the future on the basis of its (assumed to be) known past. 

\CCLsection{Constructing the Future Population-Level Birth Rate} 
\label{s:RE}

Even though it is not exactly realistic, let us assume that we know $b(t)$ up to $t=0$. More precisely, assume that
\begin{equation}\label{history-b}
b(\theta) = \phi(\theta), \quad \theta \leq 0,
\end{equation}
with $\phi$ a given non-negative integrable function (having certain additional properties specified below). The function $\phi$ describes the (assumed to be) known history of $b$, and we consider \eqref{history-b} as an \emph{initial condition} for the RE~\eqref{RE}.
We next rewrite \eqref{RE} for $t > 0$ in the form
\begin{equation}\label{RE-conv}
b = k \ast b + f,
\end{equation}
where 
the \emph{forcing function} $f$ is defined by
\begin{equation}\label{f}
f(t) = \int_t^\infty \phi(t-a)k(a) \dd a = \int_{-\infty}^0 \phi(\theta) k(t-\theta) \dd\theta,
\end{equation}
and $\ast$ denotes the usual \emph{convolution product}, defined by
\begin{equation*}
f \ast g (t):= \int_0^t f(a)g(t-a) \dd a .
\end{equation*}
We assume that both $k$ and $f$ belong to $L^1_{loc}[0,\infty)$, and note that the bilinear convolution product maps this space to itself. With an appeal to the `almost everywhere' aspect of equivalence classes of integrable functions, we leave it open whether $b(0)$ is defined as $\phi(0)$ or as $f(0)$.
We solve \eqref{RE-conv} constructively by successive approximation which, in the present context, amounts to \emph{generation expansion}:
\begin{equation}\label{generation-expansion}
b = f + k \ast f + k^{2\ast} \ast f + \cdots
\end{equation}

\begin{exercise}
\begin{enumerate}[label=(\roman*),noitemsep,nolistsep]
\item Provide a precise inductive definition of the convolution powers of $k$. 
\item Interpret the terms in the series (and by doing so, explain the terminology `generation expansion').
\item Without worrying about convergence, define the \emph{resolvent} kernel $\res$ corresponding to $k$ by
\begin{equation}\label{resolvent}
\res := \sum_{j=1}^\infty k^{j \ast},
\end{equation}
and verify that 
\begin{equation*}
\res = k \ast \res + k \quad \text{ and } \quad \res = \res \ast k + k,
\end{equation*}
while \eqref{generation-expansion} amounts to 
\begin{equation}\label{b-solution}
b = f + \res \ast f.
\end{equation}
\item \label{ex:k} Provide biological motivation for the assumption that $k$ has compact support; next, provide biological motivation for the assumption that the support of $k$ is bounded away from zero; check that the latter assumption guarantees that, for any given finite time $t$, the generation expansion has only finitely many non-zero terms (and that, consequently, it is justified that we did not worry about convergence; if you do not like this escapism from a mathematical challenge, feel free to address the issue of convergence or consult \citet{Gripenberg1990}).
\end{enumerate}
\end{exercise}

Motivated by \ref{ex:k} above, we assume from now on that $k$ has compact support, say $[0,a_{max}]$. In this case, it is useless to prescribe $b(\theta)$ for $\theta < - a_{max}$. 
Consequently, we choose $\phi$ from the positive cone in the Banach space
\begin{equation}\label{X}
X = L^1 [-a_{max},0]. 
\end{equation}
Conclusion: under very reasonable assumptions concerning the kernel $k$ and the initial history $\phi$, we can constructively define $b$ for $t > 0$ by generation expansion.

As an aside we mention that, for very special kernels $k$, the RE~\eqref{RE} is a system of ODEs in disguise. In Appendix~\ref{s:ODE} we provide a more precise and complete elaboration of this statement, as well as references to the literature.


\CCLsection{The Dynamical System Perspective} 
\label{s:dynsyst}

Time is a one-dimensional variable, and yet we shall work with two variables, viz., $t$, the time underlying the dynamical system, and $\theta$, the bookkeeping variable specifying how far back in history, relative to $t$, a birth took place. The dynamical system perspective for delay equations corresponds to making the `relative to $t$' explicit by shifting along the extended function, as expressed in the notation
\begin{equation*}
b_t(\theta) := b(t + \theta) ,  \quad \theta \leq 0. 
\end{equation*}
In words: we obtain a new function of $\theta$ by shifting over distance $t$ and we consider $b_t$ as the state at time $t$. Please note that $b$ depends on $\phi$ through \eqref{history-b} and \eqref{generation-expansion}, but that, for reasons of readability, we do not express this in the notation. The family of next-state operators $\{T(t)\}_{t\geq 0}$ on the Banach space $X$ is now defined by
\begin{equation}\label{T-semigroup}
T(t)\phi = b_t.
\end{equation}

\begin{exercise}
\begin{enumerate}[label=(\roman*),noitemsep,nolistsep]
\item Verify that $T(t)$ maps equivalence classes to equivalence classes (hint: essentially what matters here is the corresponding property for the convolution product).
\item Verify that $T(0)$ is the identity.
\item Verify that the semigroup property $T(t+s) = T(t) T(s)$, for $t,s \geq 0$, derives from the uniqueness of solutions of \eqref{RE} with \eqref{history-b} which, in turn, follows (check!) from the representation \eqref{generation-expansion}. 
\item Do you agree that the prefix `semi' is needed here? If not, delve into the difficulty of going backward in time. 
\item Verify that $\{T(t)\}_{t\geq 0}$ is strongly continuous, i.e., that 
\begin{equation*}
    \lim_{t \to 0^+} \| T(t)\phi - \phi\|_X = 0,
\end{equation*}
since translation is continuous in $L^1$ \citep{ButzerBerens1967}. 
\end{enumerate}
\end{exercise}

The \emph{infinitesimal generator} $A$ of a strongly continuous semigroup of bounded linear operators $\{T(t)\}_{t\geq 0}$ on a Banach space is the derivative of $t \mapsto T(t)$ at $t=0$, i.e.,
\begin{equation*}
    A\phi = \lim_{t \to 0^+} \frac{1}{t}( T(t)\phi - \phi), \quad \phi \in D(A) = \{ \phi \colon \text{ the limit exists}\}. 
\end{equation*}

This definition entails that, quite in general, the generator of translation of functions acts by differentiating those functions. But details concerning the domain of definition matter when we want to be more precise and specific. For instance, when the function space is $L^1(\mathbb{R})$, the domain consists of absolutely continuous functions, see \citet{ButzerBerens1967} and \citet[Appendix II]{Diekmann1995Delay}.
But when the functions are defined on a bounded interval, the precise way in which we extend them matters!

The notion of \emph{absolute continuity} of a function can be defined in multiple equivalent ways, cf.~the references above. In the present situation, there are two relevant aspects, viz., the fact that $\phi$ belongs to the domain of the generator $A$ of $\{T(t)\}_{t\geq 0}$ on $X$ defined by \eqref{T-semigroup} if and only if the \emph{extended} function is absolutely continuous, and the fact that a function is absolutely continuous if and only if it is the primitive of an $L^1$ function, so if it has a derivative represented by an element of $L^1$. 

Notation: $AC$ denotes the space of absolutely continuous functions (in the present situation these functions are defined on $[-a_{max},0]$, but we will use the same notation in other situations, trusting that the context makes it clear on what interval the functions are defined).

We are now ready to state the following result. 
\begin{theorem}
The infinitesimal generator $A$ of the semigroup $\{T(t)\}_{t\geq 0}$ on $X$ defined by \eqref{T-semigroup} is given by
\begin{equation}\label{IG}
\begin{array}{l}
A \phi = \phi', \quad \phi \in D(A), \\
\displaystyle{D(A) = \left\{\phi \in AC \colon \ \phi(0) = \int_0^{a_{max}} k(\alpha) \phi(-\alpha) \dd \alpha \right\} }. 
\end{array}
\end{equation}
\end{theorem}

Remarkably, all information about the rule for extension defined by \eqref{RE} ends up in the \emph{domain}, and NOT in the \emph{action}! Indeed, the rule shows up as the compatibility condition that $\phi(0)$ and $f(0)$ should be the same, since this guarantees that not just the initial condition $\phi$ itself is $AC$, but also its extension.

As one can imagine, having crucial information incorporated in the domain is a nuisance. Notably, changing the rule for extension amounts to an \emph{unbounded} perturbation of $A$. This technical challenge has to be overcome, one way or another. This aspect gives the dynamical system theory of delay equations its special character and, we think, charm. 
   
For linear systems of ODE, the asymptotic large time behavior of the solutions can be deduced from knowledge of the eigenvalues and eigenvectors of the matrix. Similarly one can study the asymptotic properties of $\{T(t)\}_{t\geq 0}$, for instance by performing a spectral analysis of the generator~$A$. But before doing so, we first describe an equivalent alternative dynamical system approach for describing linear age-structured population growth.

\CCLsection{The PDE Approach} 
\label{s:PDE}

The population age-density at time $t$ is an element of $L^1([0,a_{max}])$. We denote this object by $n(t,\cdot)$. In terms of the survival probability $\F$ and the population birth rate $b$, $n$ is explicitly given by
\begin{equation}\label{n}
n(t,a) = b(t-a) \F(a). 
\end{equation}

\begin{exercise}
\begin{enumerate}[label=(\roman*),noitemsep,nolistsep]
\item At time $t$, the number of individuals with age at most $a$ is given by
\begin{equation}\label{Ncumulative}
N(t,a) = \int_0^a n(t,\alpha) \dd \alpha. 
\end{equation}
These individuals are the ones that were born at most $a$ units of time ago and survived till the present time $t$. Derive \eqref{n} by combining the last two sentences.
\item Check formally that $n$ satisfies the PDE
\begin{equation}\label{PDE}
\frac{\partial n}{\partial t} + \frac{\partial n}{\partial a} = -\mu n,
\end{equation}
with boundary condition
\begin{equation}\label{BC-PDE}
n(t,0) = b(t) = \int_0^{a_{max}} \beta(\alpha) n(t,\alpha) \dd \alpha,
\end{equation}
with $\mu$ defined by \eqref{mu}. 
\end{enumerate}
\end{exercise}

The standard mathematical approach is to first use modeling considerations to derive \eqref{PDE}-\eqref{BC-PDE} and next use integration along characteristics to derive \eqref{n} with $\F$ given by \eqref{F}. Here we gave a direct biological interpretation of \eqref{n}, so there is in fact no need to introduce the PDE, we only mention it to show the connection with the standard approach. The advantage is that we avoid having to make precise in which sense the PDE is satisfied. The latter does involve a non-negligible amount of work, since several things are needed. First, one may interpret $\frac{\partial}{\partial t}+\frac{\partial}{\partial a}$ as a directional derivative. But if $b$ is not $AC$, this does not suffice, and we should probably take recourse to solutions in the sense of distributions, see, for instance, \citet{Perthame2006Book} and \citet{Webb1985Book}.

The initial condition corresponding to \eqref{PDE}--\eqref{BC-PDE} reads
\begin{equation}\label{IC-PDE}
n(0,a) = \psi (a),
\end{equation}
with $\psi$ an element of $L^1[0,a_{max}]$. 
In case of an experiment starting at $t=0$, $\psi$ may arise out of manipulations of the experimenter. But if we are describing undisturbed dynamics, we may assume that \eqref{n} also holds for $t=0$ and conclude that $\psi$ and $\phi$ (the initial condition for the RE, see \eqref{history-b}) are related by 
\begin{equation}\label{psi-phi}
\psi (a) = \phi (-a) \F(a),
\end{equation}
which we write as
\begin{equation*}
\psi = L \phi,
\end{equation*}
with $L \colon L^1[-a_{max},0] \to L^1[0,a_{max}]$ defined by $(L \phi)(a) := \phi (-a) \F(a)$. 
The correspondence between $\phi$ and $\psi$ is illustrated in Figure \ref{fig:IC}.
We assume that $\F(a_{max}) > 0$ or, in other words, that $k(a) = 0$ for $a > a_{max}$ since $\beta$ has this property, and conclude that $L$ is invertible. Note that both $L$ and $L^{-1}$ are bounded linear operators.

\begin{exercise}\label{ex:5.2}
\begin{enumerate}[label=(\roman*),noitemsep,nolistsep]
\item Using the inverse of $L$, rewrite $f$ defined in \eqref{f} in terms of $\psi$. Reinterpret $f$ as the expected rate at which the subpopulation of individuals already alive at time zero gives birth at time $t$. Check that the expression for $f$ in terms of $\psi$ is completely in line with this interpretation for general $\psi$, simply since the conditional probability to be alive at age $a+t$, given that the individual is alive at age $a$, equals $\F(a+t)/\F(a)$.
\item Check that $n(t,a)$ is, for $t \geq 0$, given by 
\begin{equation}\label{n-solution}
n(t,a) = 
\begin{cases}
    b(t-a) \F(a), & \text{if } a<t, \\[5pt]
    \psi(a-t) \displaystyle{\frac{\F(a)}{\F(a-t)}}, & \text{if } a\geq t.
\end{cases}    
\end{equation}
\item Define, for $t \geq 0$, 
\begin{equation*}
S(t) \psi = n(t,\cdot)
\end{equation*}
and check that $S(t) = L T(t) L^{-1}$. Conclude that the semigroups $\{S(t)\}_{t\geq 0}$ and $\{T(t)\}_{t\geq 0}$ describe the same dynamical system in different coordinates.
\item \label{ex:5.2-iv} Derive the infinitesimal generator of $\{S(t)\}_{t\geq 0}$.
\end{enumerate}
\end{exercise}

Conclusion: the RE dynamical system and the PDE dynamical system differ in bookkeeping details, but describe exactly the same dynamics.

Cautionary note: if we allow $a_{max} = \infty$, there are subtle aspects of behavior at infinity that thwart the simple one-to-one correspondence as incorporated in $L$ above, see \citet{Barril2022}.

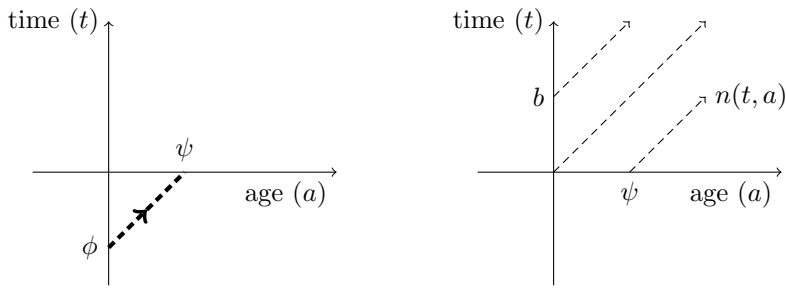
\begin{figure}[t]
    \centering
    \input{Figures/fig-IC.tex}
    \caption{Left: correspondence between the initial condition of the RE and the PDE, where the age-density $\psi$ at time $t=0$ is related to the initial birth rate history $\phi$ via \eqref{psi-phi}. Right: integration along characteristics maps $b(t)$ to $n(t,a)$ via \eqref{n-solution}.}
    \label{fig:IC}
\end{figure}


\CCLsection{Asymptotic Large Time Behavior}
\label{s:asymptotic}

First, we present a hand-waving derivation of the main results. Subsequently, we shall list various ways of providing rigorous proofs, with pointers to the relevant literature.

Equation \eqref{RE} is linear and translation invariant. This suggests to look for special solutions of the form
\begin{equation} \label{b-exponential}
b(t) = \ee^{\lambda t},
\end{equation}
with $\lambda$ to be determined. (The point is that exponential functions are precisely those functions for which translation amounts to multiplication by a constant!)

If we substitute \eqref{b-exponential} into \eqref{RE}, we obtain that $\lambda$ should satisfy the \emph{characteristic equation} 
\begin{equation}\label{CE}
\overline{k}(\lambda) = 1,
\end{equation}
named after Euler and Lotka. Here, $\overline{k}$ denotes the Laplace transform of the kernel $k$ defined in \eqref{k}. 
We assume that $k$ takes positive (nota bene: here and below we write `positive' even though `non-negative' would be more accurate) values and that $k$ is nontrivial (i.e., takes strictly positive values on a set of positive measure). 
This has important implications: if we restrict $\lambda$ to real values, then $\overline{k}$ is a strictly monotone decreasing function, with limit equal to zero for $\lambda$ tending to infinity. 
Hence \eqref{CE} has at most one real solution. 
And it does indeed have a real solution provided $\overline{k}$ assumes a value larger than one! In particular, \eqref{CE} has a solution $\gr > 0$ if and only if the \emph{basic reproduction number} $R_0$ defined by
\begin{equation}\label{R0-integral} 
R_0 := \overline{k}(0) = \int_0^\infty \beta(a) \F(a) \dd a
\end{equation}
exceeds one, see Figure \ref{fig:lotka-euler}. 
Note that $R_0$ is the expected number of daughters of a newborn individual, so the population is expected to grow (or decline, when $R_0 < 1$) with a factor $R_0$ from generation to generation.
In Appendix~\ref{s:NGO} we shall delineate a general approach for defining and computing the basic reproduction number. 

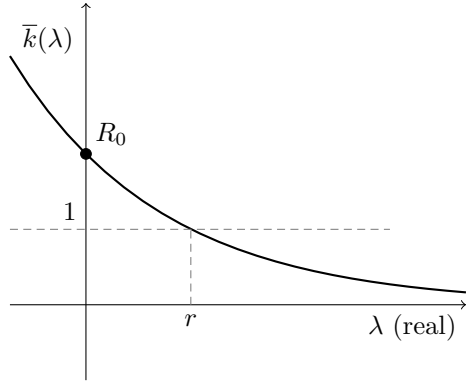
\begin{figure}[th]
    \centering
    \input{Figures/fig-lotka-euler.tex}
    \caption{Solution of the Euler--Lotka equation \eqref{CE} for $\lambda \in \mathbb{R}$. }
    \label{fig:lotka-euler}
\end{figure}
   
When $k$ has compact support, the real root $\gr$ also exists if $R_0 < 1$, since $\overline{k}(\lambda)$ tends to $+\infty$ when $\lambda$ tends to $-\infty$. For general $k$, it may happen that the tail obstructs the definition of $\overline{k}$ to the left of a line $\Re\lambda = r_c$ while $\overline{k}(r_c) < 1$; in that case, \eqref{CE} has no real solution. 

How about complex solutions of \eqref{CE}? Since
\begin{equation*}
| \overline{k}(\lambda) | \leq \overline{k}(\Re\lambda),
\end{equation*}
we know that no solution can have a real part exceeding $\gr$. But in fact there exists $\epsilon > 0$ such that for all solutions $\lambda \neq \gr$ the inequality
\begin{equation*}
\Re\lambda \leq \gr - \epsilon
\end{equation*}
holds.

\begin{exercise}
Prove this assertion.\\
Hints:  We already know that with $\epsilon = 0$ the inequality holds; use the Lemma of Riemann--Lebesgue to find a uniform upper bound for $\Im\lambda$ in a finite strip 
\begin{equation*}
\gr - \delta \leq \lambda \leq \gr 
\end{equation*}
for some $\delta > 0$;
next, observe that in a compact domain the analytic function $\overline{k} - 1$ has at most finitely many zeroes; so if we can show that there are no solutions other than $\gr$ on the line $\Re\lambda = \gr$, the existence of $\epsilon > 0$ follows; combine $\overline{k}(\gr) = 1$ with the fact that $\cos(\Im\lambda a) < 1$ for almost all $a$ if $\Im\lambda \neq 0$ to exclude solutions with $\Re\lambda = \gr$ and $\Im\lambda \neq 0$.
\end{exercise}

\begin{remark}
When starting from \eqref{RE-L}, rather than \eqref{RE}, one needs to exclude that the support of $\L$ lies in a cyclic subgroup of $\mathbb{R}$, cf.~\citet[Theorem~9.13]{Rudin1973Book}. 
Feller calls a kernel $\L$ `arithmetic' if the support does lie in such a cyclic subgroup, see Section~XI.1 and Definition~3 in Section~V.2 of \citet{Feller1971volII}. 
For such kernels, the RE \eqref{RE-L} is in fact a difference equation (in continuous time) and the characteristic equation \eqref{CE} has countably many roots on the line $\Re\lambda = \gr$. The imaginary parts of these roots form an additive group. Section~II.11 of \citet{MetzDiekmann1986}, entitled `The merry-go-round', provides an example of how such a situation can be hard to detect when one works with a PDE formulation.
\end{remark}

So far we looked at solutions of \eqref{RE} defined for all time. But in Section~\ref{s:RE} we used \eqref{history-b} to prescribe an initial history, and next defined $b$ for positive time as the unique solution of \eqref{RE-conv}.
The above observations about the roots of \eqref{CE} make it tempting to conjecture that, for $b$ defined in this manner, it holds that for some $c = c(\phi) > 0$
\begin{equation}\label{b-asympt}
b(t) \sim c \, \ee^{\gr t} \qquad \text{ for } t \to +\infty.
\end{equation}
(In words: there is a well-defined asymptotic growth rate $\gr$ that does not depend on the initial condition; the influence of the initial condition is completely captured by the scalar $c$ or, equivalently, by a translation of~$\ee^{\gr t}$.) 

We shall see below that \eqref{b-asympt} is indeed correct. And accordingly we shall interpret the identity
\begin{equation}\label{2.33}
\sign ( R_0 - 1 ) = \sign (\gr)
\end{equation}
as telling us that the generations are expected to grow if and only if the population size is expected to grow in real time. 

\begin{exercise}\label{ex:2.6.3}
Use \eqref{T-semigroup} to deduce that
\begin{equation}\label{2.34}
T(t)\phi \sim c(\phi) \ee^{\gr t} \phi_d \qquad \text{ for } t \to +\infty,
\end{equation}
with
\begin{equation*}
\phi_d(\theta) := \ee^{\gr \theta}.
\end{equation*}
(In words: the asymptotic large-time behavior of the semigroup $T(t)$ is characterized by exponential growth, with rate $\gr$, in the one-dimensional subspace spanned by $\phi_d$.)
\end{exercise}

\begin{exercise}\label{ex:2.6.4}
Use \eqref{n-solution} to deduce that, similarly, the age distribution exhibits, asymptotically for $t \to +\infty$, exponential growth with rate~$\gr$ within a one-dimensional subspace spanned by the so-called \emph{stable age distribution}~$\psi_d$ defined by
\begin{equation*}
\psi_d(a) := \ee^{-\gr a} \F(a).
\end{equation*}
\end{exercise}

\bigskip
There exists a plethora of methods to provide proofs, but in all of these `positivity' (based on order structure) is an essential ingredient. We mention:
\begin{enumerate}[noitemsep,nolistsep]
    \item Laplace transform method and Paley--Wiener Theorems: see \citet{Feller1941} and \citet{BellmanCooke1963} for early work, and \citet{Gripenberg1990} for the definitive treatment;
    \item Feller’s real analysis approach, see \citet[Chapter~XI]{Feller1971volII};
    \item The `Generalized Relative Entropy' method, see \citet{Perthame2006Book}; 
    \item Lower-bound functions and piecewise-deterministic Markov Processes, see \citet{LasotaMackey2013Book, Rudnicki2017Book}; 
    \item \label{method-5} Positive semigroup theory, based on a spectral analysis of the infinitesimal generator and on the Krein--Rutman Theorem, the infinite-dimensional variant of the Perron--Frobenius Theorem, see \citet{Arendt1986, Batkai2017, EngelNagel2000}. 
\end{enumerate}

\smallskip 
In the spirit of our focus on the RE, we here emphasize Method~1. A key point in this approach is that Laplace transformation converts the convolution product in \eqref{RE-conv} into an ordinary product. This leads to the explicit expression
\begin{equation*}
\overline{b} = ( 1 - \overline{k})^{-1} \overline{f},
\end{equation*}
and next, by inverse Laplace transformation, to a representation for $b(t)$ in the form of an improper contour integral. The singularities of the integrand are located at the points where $\overline{k}$ equals one, so at the solutions of \eqref{CE}. Next, one needs careful estimates to derive \eqref{b-asympt} by shifting the contour, see for instance \citet[Section~I.5]{Diekmann1995Delay}.
   
We also mention the Paley--Wiener result that for general (meaning: not necessarily positive and not necessarily with compact support) integrable kernels $k$, the resolvent $\res$, cf.~\eqref{resolvent}, is integrable over $[0,+\infty)$ if and only if \eqref{CE} has no solutions in the closed half-plane $\{\lambda \colon \Re\lambda \geq 0 \}$. 
This characterization of stability extends to systems of RE and was generalized by I.~M.~Gelfand to incorporate weight functions. We refer to the Theorems~2.4.1, 4.4.3 and 7.2.4 by \cite{Gripenberg1990}.

As we have seen in the Exercises \ref{ex:2.6.3} and \ref{ex:2.6.4}, \eqref{b-asympt} allows us to fully characterize the asymptotic large time behavior of the strongly continuous semigroups $\{T(t)\}_{t\geq 0}$ and $\{S(t)\}_{t\geq 0}$. In the literature, such characterizations are often derived with Method~\ref{method-5}.

\CCLsection{Widening the Framework}
\label{s:widening}
Both for the RE (with the history of the birth rate as the state) and for the PDE (with the current age distribution as the state), we chose $L^1$ as the state space, cf.~\eqref{X} and the line below \eqref{IC-PDE}. For describing aging (i.e., translation) and survival, this is a convenient choice. But when it comes to describing reproduction, problems arise. The inflow of newborns is concentrated in a point ($\theta=0$ in Section \ref{s:dynsyst} and $a=0$ in Section \ref{s:PDE}) and an element of $L^1$, being an equivalence class of functions, does not have a well-defined value in a point. In~\eqref{IG} and in Exercise~\ref{ex:5.2}\ref{ex:5.2-iv} we saw a mathematical manifestation of this difficulty: essential information is contained in the domain of the infinitesimal generator and \emph{not} in its action.

When we want to build a qualitative stability and bifurcation theory, this is a major stumbling block. As shown in detail by \cite{Diekmann2008Suns,Diekmann1995Delay}, perturbation theory for adjoint semigroups, often affectionately called \emph{sun-star calculus}, offers an efficient method to overcome this difficulty. 
In Appendix~\ref{s:sun-star}, we shall provide a concise summary of sun-star calculus. Here we shall introduce the key components in an ad hoc manner, tailored to the situation at hand. Thus we try to avoid to put readers off by the introduction of relatively heavy functional analytic machinery at a stage where the benefit is still unclear.

As it turns out, the extended framework is also very well suited for developing a numerical method based on pseudospectral approximation. So the present section serves at the same time as an introduction to our description of this method in Section~\ref{s:pseudospectral}.

\bigskip
An element $\phi$ of $L^1$ can be represented by a primitive $\Phi$, normalized to be zero in a reference point. Let us focus on \eqref{X}, and define
\begin{equation}\label{Phi}
\Phi(\theta) = - \int_{\theta}^0 \phi(\sigma) \dd\sigma.
\end{equation}

Then $\Phi$ is an element of the space of \emph{normalized bounded variation} functions, $NBV[-a_{max},0]$, where the normalization includes that $\Phi(0) = 0$ and that $\Phi$ is continuous from the right. (Incidentally, we will be a bit inconsistent when it comes to specifying the domain of the functions in the notation of a function space. Sometimes we will, sometimes we will not, the latter in particular if we feel that the context should leave little doubt about it. The range is by default $\mathbb{R}$ or $\mathbb{C}$, but please note that for systems of equations the theory works just as well.)

Please note that there exists a one-to-one correspondence between the space $NBV([-a_{max},0])$ and the space $M$ of Borel measures on the interval $[-a_{max},0]$, cf.~for instance \citet[Appendix~I]{Diekmann1995Delay}. Also, recall that a function is called absolutely continuous exactly when it is the primitive of an $L^1$ element. We define
\begin{equation*}
AC_0 := \{\Phi \colon \Phi \text{ is absolutely continuous and } \Phi(0)=0\},
\end{equation*}
and turn this space into a Banach space by defining the norm of $\Phi$ as the $L^1$-norm of its derivative, which is, in fact, the (total variation) norm of $\Phi$ as an element of $NBV$. Then, \eqref{Phi} shows that we may represent $L^1$ by~$AC_0$.

Do we gain anything by using this representation? Yes! 
When we consider $AC_0$ as a closed subspace of $NBV$, we have this bigger enveloping space at our disposal and, importantly, this bigger space is a dual space that comes naturally equipped with a second topology, the weak* topology (if these words frighten you, don’t worry and keep reading, we shall be more explicit in less technical terms soon).

We now assume that $k$ is a bounded measurable function on $\mathbb{R}_+$ with support belonging to $[0,a_{max}]$. This allows us to rewrite \eqref{f} in the form
\begin{equation}\label{f-Phi} 
f(t) = \int_{t - a_{max}}^0 k(t - \theta) \Phi(\dd\theta),
\end{equation}
where we use the Stieltjes integral notation and, also, adopt the convention that the integral equals zero when the lower integration boundary exceeds the upper integration boundary. We refer to \citet[Appendix~A]{Diekmann2021Twin} for a summary of relevant integration theory, with references to various books for precise proofs. Note that \eqref{f-Phi} makes sense for any $\Phi$ in $NBV$, not just for $\Phi$ in $AC_0$. Once $f$ is defined, we define $b$ for $t > 0$ by \eqref{RE-conv} or, in other symbols, by \eqref{b-solution}. Next, we define $B$ by
\begin{equation}\label{B-integrated}
B(t) := \int_0^t b(\sigma) \dd\sigma
\end{equation}
for $t > 0$, $B(0) = 0$ and $B(t) = \Phi(t)$ for $t < 0$.

If we use \eqref{Phi} to `translate' \eqref{T-semigroup} from $L^1$ to $NBV$, we obtain a semigroup $\{\widetilde{T}(t)\}_{t\geq 0}$ of bounded linear operators on $NBV$ defined by
\begin{equation} \label{Ttilde}
(\widetilde{T}(t) \Phi)(\theta) =  B(t+\theta) - B(t).
\end{equation}

This semigroup is \emph{not} strongly continuous (see below). Its restriction to $AC_0$ is strongly continuous. In fact, $AC_0$ is exactly the subspace consisting of elements that are starting point of a continuous orbit. 

Because of the normalization, the (Heaviside) function
\begin{equation}\label{heaviside}
H(\theta) := 
\begin{cases}
    - 1 & \text{for } \theta < 0, \\
    0 & \text{otherwise,}
\end{cases}
\end{equation}
has, when considered as an element of $NBV$, a jump of size one in $\theta = 0$. So $H$ corresponds via integration to the unit Dirac measure concentrated in $\theta=0$. The orbit starting at $H$ is \emph{not} continuous: the total variation distance between $H$ and an arbitrarily small translate of $H$ is two.

When we choose $\Phi = H$ in \eqref{f-Phi}, we find that $f=k$ and hence $b=\res$. So the resolvent $\res$ describes the population birth rate when, at time zero, we introduce a cohort (of unit size) of newborn individuals into an `empty' population.

If we formally calculate the limit, for $t$ tending to zero, of $\frac{1}{t}( \widetilde{T}(t) \Phi  -  \Phi )$, pointwise in $\theta$ for $\theta < 0$, we obtain $\Phi'(\theta) - b(0)$, where $b(0) = f(0)$ with~$f$ given by \eqref{f-Phi}. This suggests that the semigroup $\{\widetilde{T}(t)\}_{t\geq 0}$ is, in some appropriate sense, generated by an operator $C$ that is the sum of a differentiation operator $C_0$ and an operator with one-dimensional range, spanned by $H$, and as `coefficient' the map $\Phi \mapsto f(0)$. Sun-star calculus yields a rigorous proof that this is correct, cf.~\citet[Section~3.2]{Diekmann2008Suns}. Alternatively, one can use the framework of twin semigroups, see \citet{Diekmann2021Twin}. In both approaches, a key point is that on $NBV$ we have, in addition to the norm topology, a second (and weaker) topology. This topology is used when computing the generator and, even more importantly, when defining the integral in the variation-of-constants formula (recall that $\widetilde{T}(t)$ is not strongly continuous, preventing the combined use of the norm and the Riemann integral).

\begin{exercise}\label{ex:PDE-integrated}
Define $N(t,a)$ as in \eqref{Ncumulative}.
Integrate the PDE \eqref{PDE} with respect to $a$ in order to reformulate, using also \eqref{BC-PDE}, the equation as a PDE for $N$. Interpret this PDE as the abstract ODE
\begin{equation*}
\frac{\dd N}{\dd t} = A_0 N + B N,
\end{equation*}
where $A_0$ is an unbounded operator with action 
\begin{equation*}
    (A_0 \Psi)(a) := - \Psi'(a) - \int_0^a \mu(\alpha) \Psi(\dd\alpha),
\end{equation*}
while $B$ is the bounded operator
\begin{equation*}
    B \Psi := \int_0^{\infty} \beta(\sigma) \Psi(\dd\sigma) \, \widetilde{H}
\end{equation*}
with $\widetilde{H}$ the Heaviside function with $\widetilde{H}(0)=0$ and $\widetilde{H}(a) = 1$ for $a > 0$.
\end{exercise}

We end this section by explaining informally how the idea of an extended state space (cf.~\citet[Section~1.4.2]{MagalRuan2018book} and \citet[Section~6.4.2]{Inaba2017book}) relates to the approach described above.

A key mathematical difficulty is the following. The operator describing reproduction has one-dimensional (for scalar problems; in case of systems, the dimension is equal to the size of the system) range. However, the range is not contained in $X=L^1$. So we need to `enlarge' the state space $X$. When we do this in a minimal manner, we go from $X$ to $\mathbb{R}\times X$. In the context of age-structured models, the $\mathbb{R}$ component specifies the size of a Dirac measure (represented by a Heaviside function, if one works with $NBV$) in, respectively, $a=0$ or $\theta=0$. The original state space $X$ is identified with $\{0\}\times X$ and $\mathbb{R}\times X$ only serves as an auxiliary space, allowing to formulate suitable variation-of-constants formulae. A complicating factor is that the translate of the Dirac measure in a point is a Dirac measure in another point. As a consequence, translation does not map $\mathbb{R}\times X$ into itself. This difficulty is overcome by integration, more precisely, by working with integrated semigroups, cf.~\citet[Chapter~3]{MagalRuan2018book} and \citet{Thieme1990integrated}. 

So, the ideas underlying the `extended state space' approach are much the same as the ideas explained above that lead to the sun-star calculus summarized in Appendix~\ref{s:sun-star}. But one does not `enlarge' the state space quite as much and one introduces integrated semigroups in order to have a variation-of-constants formula. On the `bigger' sun-star space $M\equiv NBV$, the integrated semigroup is in fact the integral of a semigroup! This semigroup is not strongly continuous, precluding the use of the Riemann integral in the variation-of-constants formula. Since the sun-star space is, by definition, a dual space, the weak* integral comes to our rescue.

\CCLsection[Nonlinear RE and Nicholson's Blowfly Model]{Nonlinear Renewal Equations and Nicholson's Blowfly Model as a Worked Example}
\label{s:nicholson}

Gurney, Blythe and Nisbet \citep{Gurney1980} found good quantitative agreement between data from Nicholson's classic blowfly experiments and solutions of a certain delay differential equation. Their work shows that some form of density dependence can lead to narrow `discrete' generations in cycling populations. 
Here, we first formulate the underlying model as a nonlinear RE. Next, we perform a linearized stability analysis (involving the roots of a characteristic equation) of the unique trivial steady state. In the next section, we shall follow up on this with a numerical bifurcation analysis featuring periodic solutions. 

Consider a population of blowflies divided into larvae (i.e., non-\-re\-pro\-duc\-ing juvenile individuals) and adult (reproducing) individuals. Larvae become adult after a fixed maturation delay $\tau>0$, and all individuals experience mortality described by a survival probability $\F(a)$. 
Each adult individual has a reproductive potential $\beta(a)$ for $a>\tau$, and a maintenance cost proportional to it with proportionality constant $\gamma>0$. 
Scramble competition of adults limits the food supply and reduces the reproductive rate (as resources are allocated to maintenance first), so the total egg production is assumed to be a nonlinear function $h$ of the total adult reproductive potential, with the property that it tends to zero as the number of adults increases. 


In particular we define 
\begin{equation*}
    h(x) = x \ee^{-\gamma x},
\end{equation*}
and assume that, at the $p$-level ($p$ for `population'), the egg production rate is given by
\begin{equation}\label{Nicholson1}
    b(t) = h(\tilde{b}(t)), \qquad t>0, 
\end{equation}
where
\begin{equation}\label{Nicholson2}
    \tilde{b}(t) = \int_\tau^{a_{max}} \beta(a) \F(a) b(t-a) \dd a.
\end{equation}
In contrast to the framework described in Section~\ref{s:constant_env}, the environmental condition is determined by the adult population maintenance costs, and hence is not constant: \eqref{Nicholson1} with \eqref{Nicholson2} is a \emph{nonlinear} RE.

Nontrivial equilibria of \eqref{Nicholson1} are constant functions $b(t) \equiv b^*>0$ such that 
\begin{equation*}
    1 = R_0 \, \ee^{-\gamma R_0 b^* }, 
\end{equation*}
where the basic reproduction number $R_0$ is defined by 
\begin{equation*}
    R_0 = \int_\tau^{a_{max}} \beta(a) \F(a) \dd a.
\end{equation*}
Note that $R_0$ specifies the expected total lifetime number of offspring of a newborn individual in the low density limit, i.e., when the influence of density dependence can be ignored. 
Hence, a unique nontrivial equilibrium exists when $R_0>1$. 

For a general nonlinear RE 
\begin{equation}\label{RE-nonlinear}
    b(t) = F(b_t), \qquad t>0,
\end{equation}
where $F \colon X \to \mathbb{R}$ is at least continuously differentiable and $X$ is defined in \eqref{X}, the \emph{principle of linearized stability} holds, ensuring that the local stability of an equilibrium $b^*$ is determined by the stability of the zero solution of the linearized equation 
\begin{equation}\label{RE-linearized}
    x(t) = DF(b^*) \, x_t, \qquad t>0, 
\end{equation}
where $DF(b^*)$ is the Fr\'echet derivative of $F$ computed at $b^*$ \citep{Diekmann2008Suns}. Note that by the Riesz representation theorem, \eqref{RE-linearized} can be written as the RE
\begin{equation*}
    x(t) = \int_0^\infty x(t-a) k(a) \dd a
\end{equation*}
for a suitable kernel $k \in L^\infty$ with compact support. Note that $k$ may take negative values.  
The stability of the zero solution of \eqref{RE-linearized} is then determined by the roots of the Euler--Lotka equation \eqref{CE} as described in Section \ref{s:asymptotic}:\\
\emph{(i)} if all roots have strictly negative real part, $b^*$ is asymptotically stable;\\
\emph{(ii)} if at least one root has strictly positive real part, $b^*$ is unstable. 

A special case is obtained by taking the limit that reproduction occurs by way of a burst upon reaching adulthood, i.e., in the limit $a_{max} \searrow \tau$ and $\beta \rightarrow \infty$, with $\int_\tau^{a_{max}} \beta(a)\F(a)\dd a$ converging to a finite positive value, leading to a neutral RE. 
For a theoretical framework to deal with neutral equations we refer to \citet{Diekmann2021Twin}, but note that this paper only deals with linear equations. 

If $k$ does not have compact support, the principle of linearized stability holds in an exponentially weighted $L^1$ state space, and the stability of the zero solution of the linearized equation is determined by the characteristic roots in a right-half complex plane ($\Re\lambda > -r_c$ for some $r_c > 0$) where the Laplace transform $\overline{k}(\lambda)$ is well defined \citep{Diekmann2012Blending}. 

\begin{exercise}
Show that the linearization of \eqref{Nicholson1} around the nontrivial equilibrium $b^*>0$ reads
\begin{equation*}
    x(t) = \frac{\left( 1 - \ln R_0 \right)}{R_0} \int_\tau^{a_{max}} \beta(a) \F(a) x(t-a) \dd a,
\end{equation*}
and the corresponding characteristic equation for $\lambda \in \mathbb{C}$ is 
\begin{equation}\label{CE-Nicholson}
    1 = \frac{1 - \ln R_0}{R_0} \int_\tau^{a_{max}} \beta(a) \F(a) \ee^{-\lambda a} \dd a.
\end{equation}
\end{exercise}

\begin{exercise}
Consider the case $a_{max}=\infty$, $\beta(a) \equiv \beta$ when $a>\tau$, and $\F(a) = \ee^{-\mu a}$ for $\mu>0$.
\begin{enumerate}[label=(\roman*),noitemsep,nolistsep]
\item Convince yourself that the total number of adults $N(t)$ is given by
\begin{equation*}
    N(t) = \int_\tau^\infty \ee^{-\mu a} b(t-a) \dd a,
\end{equation*}
then show that $N(t)$ satisfies the delay differential equation (DDE) considered by \citet{Gurney1980}, 
\begin{equation}\label{Nicholson-DDE}
    \frac{\dd N(t)}{\dd t} = - \mu N(t) + \beta \ee^{-\mu\tau} h(N(t-\tau)).
\end{equation}

\item Show that \eqref{CE-Nicholson} simplifies to
$$ \lambda - b_1 - b_2 \ee^{-\lambda} = 0,$$
with
$$ b_1 =  -\mu + \frac{\mu(1-\ln R_0)}{1-\ee^{-\mu}} , \qquad b_2 = \frac{\mu \ee^{-\mu} (1-\ln R_0)}{1-\ee^{-\mu}}.$$
See also \citet[Section~2]{DeWolff2021}. 
\end{enumerate}
\end{exercise}

\begin{exercise}
In the case of burst reproduction, show that \eqref{CE-Nicholson} simplifies to
$$ 1 = (1-\ln R_0) \, \ee^{-\lambda \tau},$$
and verify that, for $R_0=1$, there are countably many roots on the imaginary axis (forming a group under addition). 
\end{exercise}

Sometimes, for instance when $a_{max}=\infty$ or $a_{max} \searrow \tau$, one can describe analytically how the parameter space partitions in a stability and an instability region \citep{Diekmann2013didactical}. The boundary between these regions has, in general, two types of components: zero eigenvalue (the multiplicity of steady states changes) or a pair of imaginary eigenvalues (corresponding to a Hopf bifurcation). 
We refer to \citet[Appendix~A]{DeWolff2021} for the analytical study of stability of~\eqref{Nicholson-DDE}. 

When the analytical description is, for whatever reason, not possible, one can still try to compute the stability boundary in parameter space by numerical continuation: we refer for instance to the software package MatCont for the numerical continuation of curves in MATLAB \citep{MatCont2008}.

Studying analytically the dynamical and bifurcation properties far from the Hopf bifurcation curve (e.g., periodic solutions and their stability) is usually impossible, so numerical bifurcation tools become of crucial importance.
The lack of ready-to-use bifurcation software has most likely hampered the use of RE models in applications. 

In the next section, we introduce the pseudospectral method to study numerically the bifurcations of general nonlinear equations \eqref{RE-nonlinear} by means of approximating systems of ODE, for which a much larger pool of established software packages exists, and we illustrate its potential by analyzing the Nicholson blowfly model. 

For bifurcation theory for DDE in general we refer to the book by \cite{GuoWu2013}.

\CCLsection{Numerical Bifurcation Analysis via Pseudospectral Approximation}
\label{s:pseudospectral}

With the nonlinear RE \eqref{RE-nonlinear} we have associated a dynamical system on the infinite-dimensional function space $X$. We now want to approximate elements of $X$ by elements of a finite-dimensional subspace, and next derive a system of ODE to describe the dynamics of the approximation. 
This reduction can be done in various ways. 
We here focus on the spectral collocation method, also called \emph{pseudospectral method}, see for instance \citet{BMVbook} and \citet{Ando202215years}. 

We summarize the approach for nonlinear RE adopted by \citet{Scarabel2021JCAM}, to which we refer for further details. It considers the integrated states~\eqref{Phi}. For simplicity we present the method in the scalar case, although it can be easily applied to systems of equations. 
Besides the already mentioned advantage that the rule for extension appears as a bounded perturbation of the infinitesimal generator corresponding to trivial extension, working with the integrated state allows us to work with absolutely continuous functions, meaning that point evaluation and polynomial interpolation are well defined. 

We start from a general nonlinear RE~\eqref{RE-nonlinear}.
Through the integrated variable~\eqref{B-integrated}, we can define the nonlinear analogue of~\eqref{Ttilde} and consider the state $v(t)$ such that 
\begin{equation*}
    v(t)(\theta) = B(t+\theta)-B(t).
\end{equation*}
As described at the end of Section~\ref{s:widening}, $v(t)$ satisfies $v(t)(0)=0$ and the abstract differential equation
\begin{equation}\label{ADE-NBV}
    \frac{\dd v(t)}{\dd t} = C_0 v(t) + F(C_0 v(t)) \cdot H,
\end{equation}
where $C_0$ is the differentiation operator $C_0 \Phi = \Phi'$, for $\Phi\in AC_0$, and $H$ is the Heaviside function~\eqref{heaviside}. 
\begin{exercise}
    Derive \eqref{ADE-NBV} by combining \eqref{B-integrated} and \eqref{RE-nonlinear}.
\end{exercise}

In the pseudospectral approach, $v(t) \in AC_0$ is approximated by a polynomial in~$\theta$, such that this polynomial satisfies \eqref{ADE-NBV} on a finite set of points in $[-a_{max},0]$, called \emph{collocation nodes}. This will lead to a finite-dimensional system of ODE in which each variable describes the time-evolution of the value of the approximating polynomial in a node. 

\begin{figure}[tb]
    \centering
    \input{Figures/fig-interpolation.tex}
    \caption{Illustration of the polynomial interpolation process described by~\eqref{pM}.}
    \label{fig:interpolation}
\end{figure}
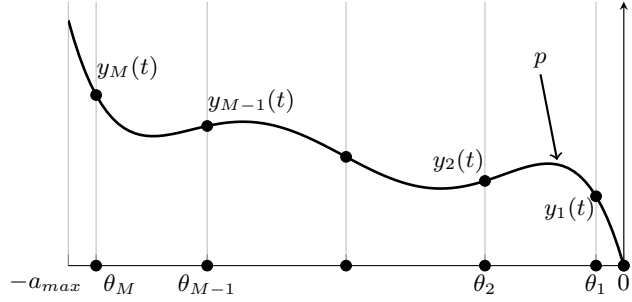

\bigskip
Given an integer $M>0$, the \emph{discretization index}, we consider a set of $M$ distinct points $\{\theta_1,\dots,\theta_M\}\subset [-a_{max},0)$, ordered such that
$$ -a_{max} \leq \theta_M < \theta_{M-1} < \cdots < \theta_1 < 0.$$
We introduce $M$ variables $y_j(t)$, $j=1,\dots,M$, and, for each $t\geq 0$, the $M$-degree polynomial $p$ such that
$$ p(0) = 0 \qquad \text{and} \qquad p(\theta_j) = y_j, \quad j=1,\dots,M.$$
Note that $p$ depends on time~$t$, since $y_j$ does. 
According to standard approximation theory, $p$ can be written in the Lagrange representation as a linear combination of the $y_j$,
\begin{equation}\label{pM}
    p(\theta) = \sum_{j=1}^M y_j \ell_j(\theta),
\end{equation}
where the Lagrange polynomials, $\ell_j$, are defined as
\begin{equation*}
    \ell_j(\theta) = \prod_{\substack{i=0\\ i\neq j}}^M \frac{\theta-\theta_i}{\theta_j-\theta_i}, \qquad j=0,\dots,M,
\end{equation*}
and form a basis of the space of the $M$-degree polynomials \citep[Section~2.5]{Davis1975book}. 

Next we require that the polynomial~\eqref{pM} satisfies~\eqref{ADE-NBV} on the collocation nodes, i.e., for $\theta = \theta_1,\dots,\theta_M$. 
By linearity,
$$ \frac{\partial p}{\partial t} = \sum_{j=1}^M \frac{\dd y_j}{\dd t} \ell_j, \qquad \text{and} \qquad  \frac{\partial p}{\partial \theta} = \sum_{j=1}^M y_j \frac{ \dd\ell_j}{\dd \theta}.$$
Hence, using the property $\ell_j(\theta_i) = \delta_{ij}$, where $\delta_{ij}$ is Kronecker's delta, we obtain the $M$-dimensional system of ODE
\begin{equation}\label{ODE}
    \frac{\dd y_i}{\dd t} = \sum_{j=1}^M y_j \ell_j'(\theta_i) - F(\sum_{j=1}^M y_j \ell_j'), \qquad i=1,\dots,M. 
\end{equation}
Given the solution $y_1,\dots,y_M$ of \eqref{ODE}, an approximation of the solution $b(t)$ of \eqref{RE-nonlinear} is given by $F(\sum_{j=1}^M y_j \ell_j')$. 

System \eqref{ODE} can be used to study the stability of equilibria of \eqref{RE-nonlinear} and their bifurcations in the sense specified by \citet{Scarabel2021JCAM}. In particular: the equilibria of \eqref{RE-nonlinear} and \eqref{ODE} are in one-to-one correspondence; linearization and pseudospectral approximation commute; and a given characteristic root of the linearized RE is approximated by a sequence of characteristic roots of the approximating system as $M\to \infty$, when the collocation nodes are taken as the Chebyshev zeros, i.e., the roots of the Chebyshev polynomials of the first kind \citep[Section~1.5]{Gautschi2004book}.
Since the eigenfunctions of the linearized operator are exponentials, the order of convergence to the eigenvalues (and corresponding eigenfunctions) is exponential in $M$, a behavior known as \emph{spectral accuracy} \citep[Chapter~4]{Trefethen2000spectral}\footnote{The word `spectral' here does not relate to the spectrum of an operator!}.
The use of Chebyshev nodes also allows to have explicit formulas for the elements $\ell_j'(\theta_i)$ of the differentiation matrix and for the corresponding quadrature rules, which can be computed efficiently \citep[Section~3.1]{Gautschi2004book}. 
Note also that the structure of the ODE system \eqref{ODE} reflects the structure of the abstract differential equation \eqref{ADE-NBV}, with a linear differentiation part and a perturbation with a one-dimensional range spanned by the vector $(-1,\dots,-1)^T$ that approximates the Heaviside function \eqref{heaviside}. 

A similar approach can be used to study DDE with bounded delay \citep{Breda2016Prospects}, RE and DDE with infinite delay \citep{Gyllenberg2018, Scarabel2024Infinite}, and the PDE formulation introduced in Section~\ref{s:asymptotic} \citep{Scarabel2021Vietnam}, for which we refer also to \cite{DiekmannScarabelSize}. The method has also been applied by \cite{Ando2022pseudospectral} to study the stability of linear models with two structuring variables.

In the presence of discontinuities in the model coefficients, for instance when the fertility rates have jumps at the maturation age, a piecewise approach is more powerful. The function is approximated with a polynomial in each sub-interval (potentially with varying degree, depending on the length of each sub-interval), see for instance \citet[Section 5.2]{BMVbook}. 

The numerical approximation of DDE and RE with bounded delay has recently been made available via an import utility in the graphical user interface of MatCont (version 7p6). 
The standard graphical user interface of MatCont can then be used for the numerical bifurcation analysis of the approximating ODE system obtained via pseudospectral approximation \citep{Liessi-matcont}. 
The authors and their collaborators hope that the availability of easy-to-use software will promote a wider use of age-structured models in applications.

\bigskip

\begin{figure}[pt]
    \centering
    \input{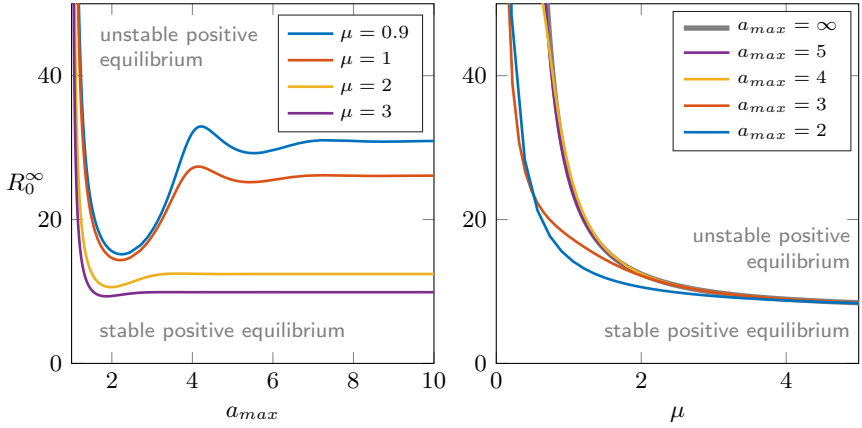}
    \caption{Stability boundaries of the positive equilibrium of \eqref{Nicholson-scaled}, for different values of $\mu$ and $a_{max}$: the equilibrium exists for $R_0>1$ (i.e., $R_0^\infty > 1+\frac{\beta}{\mu}\ee^{-\mu a_{max}}$) and is stable below the curves. The stability boundary for $a_{max}=\infty$ is computed analytically.}
    \label{fig:nicholson}
\end{figure}

\begin{figure}[pt]
    \centering
    \input{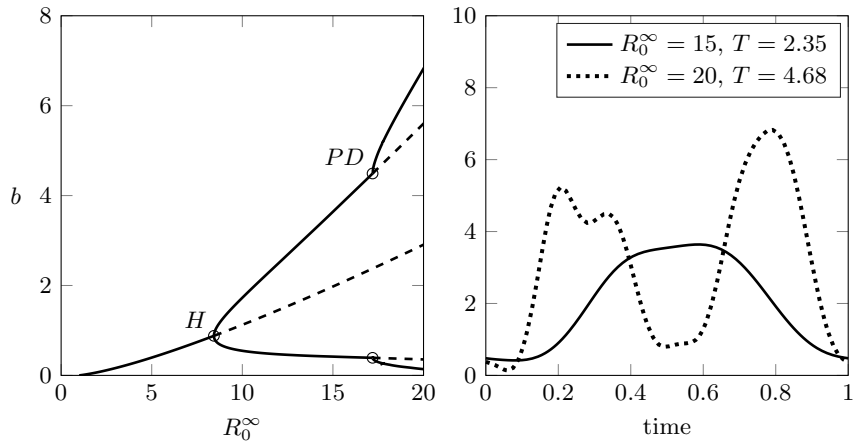}
    \caption{Left: bifurcation diagram of \eqref{Nicholson-scaled}, with $a_{max}=\mu=5$ and varying $R_0^\infty$, showing a stable branch of periodic orbits (minimum and maximum values plotted) starting from Hopf (H), which becomes unstable through a period doubling (PD) bifurcation. Right: stable periodic orbits corresponding to $R_0^\infty=15$ (before PD) and $R_0^\infty=20$ (after PD). The periods are normalized to 1, and the estimated period $T$ is shown in the legend.}
    \label{fig:nicholson-2}
\end{figure}

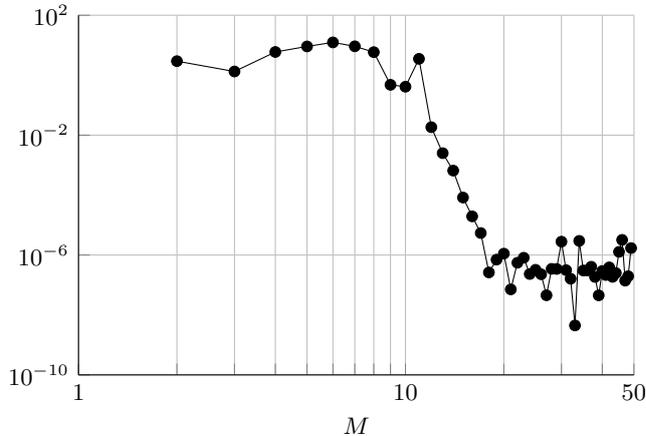
\begin{figure}[t]
    \centering
    \input{Figures/fig-error.tex}
    \caption{Log-log plot of the error in the detection of the Hopf bifurcation of \eqref{Nicholson-scaled}, with $a_{max}=\mu=5$ and varying $M_a$ (number of points in the discretization of the adult interval), for $M_j=5$ fixed. }
    \label{fig:nicholson-error}
\end{figure}

To illustrate the use of the pseudospectral approximation in applications, we consider again the Nicholson blowfly RE \eqref{Nicholson1} with constant parameters $\beta(a) \equiv \beta$, for $a \geq \tau$, and $\mu(a) \equiv \mu$.
By rescaling $b$ and $t$, without loss of generality we can set $\gamma=\tau=1$ and restrict our attention to the equation
\begin{equation} \label{Nicholson-scaled}
    b(t) = \beta \tilde{b}(t) \ee^{-\tilde{b}(t)}, \qquad \tilde{b}(t) = \int_1^{a_{max}} \ee^{-\mu a} b(t-a) \dd a. 
\end{equation}
The unique positive equilibrium $b^* = \frac{\beta \ln R_0}{R_0}$ exists when $R_0>1$, where $R_0 = \frac{\beta}{\mu}(\ee^{-\mu}-\ee^{-\mu a_{max}})$. 
Note that $R_0 \to R_0^\infty = \frac{\beta}{\mu}\ee^{-\mu}$ when $a_{max}\to \infty$. 

Figure~\ref{fig:nicholson} shows the stability boundaries of the equilibrium~$b^*$ in the $(a_{max},R_0^\infty)$ plane for various values of $\mu$, and in the $(\mu,R_0^\infty)$ plane for various values of $a_{max}$, obtained using MatCont for the pseudospectral approximation of \eqref{Nicholson-scaled}. To treat more efficiently the discontinuity in the fertility rate $\beta$, here we used a piecewise discretization using a polynomial of degree $M_j=5$ in the juvenile age-interval $[0,1]$ and a polynomial of degree $M_a=20$ in the adult age-interval $[1,a_{max}]$. Note how, as $a_{max}\to \infty$, we recover the analytically known stability boundary from \citet{DeWolff2021} in the right panel of Figure~\ref{fig:nicholson}. 

Figure~\ref{fig:nicholson-2} (left) shows the bifurcation diagram for $a_{max}=\mu=5$, varying~$R_0^\infty$. The positive equilibrium undergoes a Hopf bifurcation at $R_0^\infty \approx 8.43$, and a branch of stable periodic solutions is born. The branch of periodic solutions undergoes a period doubling bifurcation at $R_0^\infty \approx 17.19$. Two periodic solutions, one at $R_0^\infty =15$ (before the period doubling point) and one at $R_0^\infty =20$, are plotted in the right panel of Figure~\ref{fig:nicholson-2}. In the figure, the two profiles are plotted in one period, normalized to 1. The estimated period is included in the legend. These periodic orbits show the separation of reproduction into different `generations', with large peaks that can also differ in size, as first investigated by \cite{Gurney1980}. 

Figure~\ref{fig:nicholson-error} shows the behavior of the approximation error of the Hopf point, when the number of discretization points in the juvenile interval is fixed at $M_j=5$ and the number of points in the adult interval is varied from $M_a=1$ (no Hopf point detected) to $M_a=50$. The figure shows that $M_a=20$ is sufficient to approximate the Hopf bifurcation point to an accuracy of $10^{-6}$, which is the default tolerance in MatCont. The reference value for the computation of the error is taken as the value computed with $M_a=50$. The convergence behavior agrees with the theoretical results concerning the approximation error of the characteristic roots proved by \cite{Scarabel2021JCAM}.

\CCLsection{Concluding Remarks}\label{s:conclusions}

The general aim of structured population models is to incorporate the development of individuals as a component of population change, the other components obviously being the birth and death of individuals. The simplest situation arises when we just keep track of `age', i.e., time elapsed since birth. The trivial observation that an individual of age $a$ at time $t$ must have been born at time $t-a$ forms the basis for deriving straightforward bookkeeping equations capturing the population dynamics. It is easy to prove that these equations admit unique solutions once we prescribe appropriate initial conditions. As always, the operators mapping the initial data to the corresponding information at a later time form a dynamical system (in the linear case we usually speak about `semigroups of operators' rather than `dynamical systems').

To develop local stability and bifurcation theory for these dynamical systems is far less easy. The reason is that, in some sense, the ideal state space does not exist. When we choose the `small' space $L^1$, the birth term, being concentrated in a point, is singular. When we work with the `large' space of measures, translation is not strongly continuous. As we indicated in Section~\ref{s:widening}, and shall show in a bit more detail in Appendix~\ref{s:sun-star}, one way out of this difficulty is to work with both. It involves a bit of work and a bit of abstract thinking to set up a framework, but, once this is done, standard proofs can be used to develop the local qualitative theory, allowing to draw conclusions about dynamical behavior from information about the position of the roots of a characteristic equation in the complex plane, and how this position changes when parameters are varied.

For simple models, the stability (of steady states) boundary in parameter space can be determined by a pen and paper analysis. The Hopf bifurcation theorem yields information about periodic solutions close to certain parts of these stability boundaries. But in order to investigate the dynamics beyond a neighborhood of the stability boundary, one needs numerical methods. In this chapter we advocate the use of pseudospectral approximation to derive an $M$-dimensional system of ODEs that is amenable to a numerical bifurcation analysis with standard and (relatively) user-friendly tools. The dimension $M$ is a tunable parameter, controlling the trade-off between accuracy and computational costs. Often one starts with a quick crude analysis for a small value of $M$ and next explores in more detail certain domains in parameter space with higher values of $M$. This works, in a sense, as an interactive version of mesh refinement.

We surmise that often modelers shy away from age-structured models since they are uncertain about both the state of the theoretical art and of the state of the computational art. We very much hope that the present chapter contributes to diminishing this kind of uncertainty.

\appendix

\CCLsection[When Is a RE a System of ODE in Disguise?]{When Is a Renewal Equation a System of ODE in Disguise?
(with Special Attention to the Kermack--McKendrick Epidemic Model)
}
\label{s:ODE}
The assumption that the kernel $k$ in~\eqref{k} has compact support facilitates both the development of theory and the construction of numerical methods for~\eqref{RE}. And, importantly, there are convincing modeling arguments to support this assumption.

The aim of the present section is to reveal that a certain class of kernels with \emph{unbounded} support eases the task of the analyst in a totally different way. Indeed, for this class we derive an equivalent system of ODE, making both the gigantic ODE theory and the very large numerical toolbox available for analysis of the equation. 
The results disclose, in addition, the connection between familiar compartmental epidemic models and the, unfortunately less well-known, Kermack--McKendrick nonlinear RE. 
Here the helpfulness actually also works in the other direction: results that are easy to derive for the scalar RE can be translated into properties of solutions of (possibly high dimensional) systems of ODE.
We point to \citet{Cassidy2022numerical} and \citet{Cassidy2025using} for alternative representations via systems of ODE in the context of differential equations with distributed delays and for efficient numerical methods for time integration. 

The kernels for which this reduction works are matrix exponentials sandwiched between two vectors. 
More precisely, let $d$ be a positive integer, let $U,V$ be two $d$-vectors, $\Sigma$ be a $d\times d$-matrix, and consider the RE \eqref{RE}--\eqref{k} with kernel $k$ given by
\begin{equation} \label{k-exp}
k(a) = U \cdot \ee^{a\Sigma} \, V.
\end{equation}
Note that one can always normalize such that $V$ has $l_1$-norm one, since $V$ and $U$ occur multiplicatively in \eqref{k-exp}. This is convenient to interpret $V$ as a probability vector. 

Substituting \eqref{k-exp} into \eqref{RE}, we find that
\begin{equation}\label{b=Uy}
b(t) = U\cdot y(t),
\end{equation}
with the time-dependent $d$-vector $y$ defined by
\begin{equation}\label{y}
y(t) = \ee^{t\Sigma} \int_{-\infty}^t \ee^{-\theta\Sigma} \, V b(\theta) \dd\theta.
\end{equation}
By differentiating \eqref{y} with respect to $t$, we infer that $y$ satisfies the linear ODE
\begin{equation}\label{ODE-y}
\frac{\dd y}{\dd t} = \Sigma \, y + b \, V = \Sigma y + (U\cdot y) \, V.
\end{equation}

Conversely, a solution of \eqref{ODE-y} defined on $(-\infty,t]$ has, by the variation-of-constants formula, to satisfy the identity \eqref{y} with $b$ given by \eqref{b=Uy}. And by applying $U$ to \eqref{y}, we recover \eqref{RE} with kernel given by \eqref{k-exp}. See \cite{Diekmann2018finite} for general results concerning the equivalence of delay equations and systems of ODE.

One can interpret \eqref{ODE-y} in terms of a Markov process with transition matrix $\Sigma$. The vector $U$ describes the (scalar) output. The vector $V$ spans a one-dimensional subspace and the input is contained in this subspace. In fact, the output is fed back instantaneously as input. 

\bigskip
As a concrete example, consider a situation where newborn individuals are juveniles which have to mature to become adults, i.e., capable of reproduction. Assume that the length of the juvenile period is exponentially distributed with parameter $\sigma$, that all juveniles become adult (so for simplicity we ignore that juveniles can die before reaching adulthood), and that adults die with probability per unit of time $\mu$. Let $y_1$ describe the size of the juvenile subpopulation and $y_2$ the size of the adult subpopulation. For $\Sigma$ we take the $2\times 2$ matrix
\begin{equation*}
\begin{pmatrix}
    - \sigma & 0 \\ 
    +\sigma & - \mu 
\end{pmatrix},
\end{equation*}
for $V$ the 2-vector $(1,0)^T$, and for $U$ the 2-vector $(0,\beta)^T$, with $\beta$ the rate at which an adult produces offspring. The corresponding ODE reads
\begin{equation*}
\frac{\dd}{\dd t}
\begin{pmatrix}
    y_1 \\ y_2
\end{pmatrix}
=
\begin{pmatrix}
    - \sigma y_1   +  \beta y_2 \\
    \sigma y_1   -  \mu y_2
\end{pmatrix}.
\end{equation*}

For life histories with several life stages, the same strategy works. The key point is that sojourn times are exponentially distributed and that the distribution of states-at-birth, as described by the probability vector $V$,
is independent of the state of the mother. 

Incidentally, we refer to \cite{Diekmann2018finite} for a generalization of the equivalence to situations where we deal with a system of RE and~$V$ and~$U$ are matrices, so consist of multiple vectors.

\bigskip
For epidemic models the situation is similar, but with a bookkeeping and a mathematical difference:
\begin{itemize}[noitemsep,nolistsep]
    \item bookkeeping: replace `age' by `age of infection', or, in other words, consider `becoming infected' as the `birth' event, at which the Markov process starts;
    \item mathematical: the linear RE \eqref{RE} describes population growth in a constant environment, but the infection process diminishes the availability of susceptible individuals; taking the point of view of the pathogen, we consider the size $S$ of the subpopulation of susceptible individuals as the (one-dimensional) environmental condition, and this condition is affected by feedback (when we ignore demographic turnover and loss of immunity, this size can only go down).
\end{itemize}

\bigskip
The \citet{KermackMcKendrick1927} model is summarized by the two equations
\begin{equation}\label{dSdt}
\frac{\dd S}{\dd t}  =  - F S
\end{equation}
and
\begin{equation}\label{FOI}
F(t)  =  \int_0^{\infty} A(\tau) F(t-\tau) S(t-\tau) \dd\tau.
\end{equation}
Here $F$ is the \emph{force of infection}, i.e., the probability per unit of time for a susceptible to become infected (these words correspond exactly to what \eqref{dSdt} expresses mathematically). The function $A$ of the time $\tau$ elapsed since infection took place is the one-and-only model ingredient. It specifies the expected contribution of an infected individual to the force of infection. 

Now assume that
\begin{equation*}
A(\tau) = U\cdot \ee^{\tau \Sigma} \, V,
\end{equation*}
then we can rewrite \eqref{FOI} as
\begin{equation}\label{F=Uy}
F = U\cdot y
\end{equation}
with 
\begin{equation*}
y(t) := \ee^{t \Sigma} \int_{-\infty}^t \ee^{- \theta \Sigma} \, V F(\theta) S(\theta) \dd\theta.
\end{equation*}
Accordingly, $y$ satisfies
\begin{equation}\label{ODE-Kermack}
\frac{\dd y}{\dd t}  = \Sigma \, y  +  (F S)\, V,
\end{equation}
and, together, \eqref{dSdt} and \eqref{ODE-Kermack}, with $F$ defined by \eqref{F=Uy}, form a closed nonlinear system of ODE.

\begin{exercise}
Derive, starting from this system of ODE, a scalar nonlinear RE for the cumulative force of infection
\begin{equation*}
    x(t) := \int_{-\infty}^t F(\sigma) \dd\sigma.
\end{equation*}
Hint: first show that \eqref{FOI} holds for $F$ defined by \eqref{F=Uy}. Next integrate the identity \eqref{FOI}, but with $FS$ at the right hand side replaced, on account of \eqref{dSdt}, by $-S'$. Thus you should arrive at the RE
\begin{equation}\label{RE-FOI}
x(t) = \int_0^\infty k(\tau) g(x(t-\tau)) \dd\tau,
\end{equation}
with $k(\tau) := S(-\infty) A(\tau)$ and $g(x) := 1 - \ee^{-x}$. Note that $S(-\infty) = N$, the total population size, if we describe an outbreak in a virgin, i.e., completely susceptible, population.
\end{exercise}

\begin{exercise}
Specify $\Sigma$, $U$ and $V$ for the so-called SEIR (susceptible-exposed-infectious-removed) model.\\
Hint: The capital letters $S$, $E$, $I$, $R$, refer to the possible states of individuals, but are at the same time used to denote the size of the subpopulation of individuals having that state. The vector $y$ corresponds to $(E,I)^T$ (the size $R$ of the subpopulation of removed individuals is ignored in the bookkeeping scheme, since it has no impact on the dynamics). Do you see an analogy with the juvenile-adult population model discussed above? See \citet[Section 9.3]{Diekmann2018finite} for details.
\end{exercise}

For the scalar nonlinear RE \eqref{RE-FOI}, one can rather easily define the basic reproduction number $R_0$ and the Malthusian parameter~$\gr$. In addition, one can characterize the final size by an algebraic equation. These results have immediate implications for the system \eqref{dSdt} and \eqref{ODE-Kermack} with \eqref{F=Uy}, for all positive-off-diagonal matrices $\Sigma$ with negative spectral bound, all non-negative vectors $U$ and all probability vectors~$V$. It appears that the RE is a powerful tool for analysing large classes of epidemic models. We refer to the textbook by \citet{Diekmann2013EpiBook} for more background on epidemic models and to the recent survey papers by \citet{Bootsma2018} and \citet{Diekmann2025} for the infinite-dimensional generalization of the RE~\eqref{RE-FOI} and the notion of \emph{herd immunity threshold}.


\CCLsection{Sun-star calculus}\label{s:sun-star}
In this section we sketch $\odot *$ calculus by focusing on RE with finite delay. 
This is based on \citet{Diekmann2008Suns}. 
For a systematic exposition in the context of DDE we refer to the textbook by \citet{Diekmann1995Delay} and, throughout this section, we include explicit pointers to sections are references in this textbook. 
For the extension to infinite delay, see \citet{Diekmann2012Blending}. 

Let $X=L^1([-1,0],\mathbb{R})$, but note that when doing spectral analysis one should replace $\mathbb{R}$ by $\mathbb{C}$, see Section~III.7. 
We begin by considering the trivial rule for extension, meaning that we extend by zero:
\begin{equation*}
    (T_0(t)\phi)(\theta) = 
    \begin{cases}
        \phi(t+\theta) & \theta\leq -t,\\
        0 & \theta > -t,
    \end{cases}
\end{equation*}
with infinitesimal generator
\begin{equation*}
    A_0\phi=\phi', \qquad D(A_0) = \{ \phi\in AC \colon \phi(0)=0\}.
\end{equation*}
As usual, we represent the dual space by $L^\infty$, but we choose $[0,1]$ rather than $[-1,0]$ as the interval, so we take $X^* = L^\infty([0,1],\mathbb{R})$ with pairing
\begin{equation*}
    \langle g,\phi \rangle = \int_0^1 g(\sigma) \phi(-\sigma) \dd\sigma.
\end{equation*}
The dual semigroup is then 
\begin{equation*}
    (T^*_0(t)g)(\sigma) = 
    \begin{cases}
        g(\sigma+t) & \sigma\leq 1-t,\\
        0 & \sigma > 1-t,
    \end{cases}
\end{equation*}
with 
\begin{equation*}
    A^*_0\,g=g', \qquad D(A^*_0) = \{ g\in Lip \colon g(1)=0\}.
\end{equation*}
Note that $\{T_0(t)\}_{t\geq 0}$ is strongly continuous, but $\{T^*_0(t)\}_{t\geq 0}$ is not (translation is continuous in $L^1$, but not in $L^\infty$, as shown by considering an equivalence class that contains a step function). 

We denote the maximal subspace of strong continuity by $X^\odot$, so 
\begin{equation*}
    X^\odot := \{ g\in X^* \colon \lim_{t\to 0^+} \| T^*_0(t)g-g\|=0 \},
\end{equation*}
and note that in general $X^\odot=\overline{D(A^*_0)}$ (see Proposition~3.8 in Appendix~II).
In this particular case, we have
\begin{equation*}
    X^\odot = C_0([0,1),\mathbb{R})= \{g \in C([0,1),\mathbb{R}) \colon g(1)=0 \}.
\end{equation*}

The restriction $T^\odot_0(t) = T^*_0(t)|_{X^\odot}$ defines a strongly continuous semigroup on $X^\odot$, and its generator is the part of $A^*_0$ in $X^\odot$, i.e., 
\begin{equation*}
    D(A^\odot_0)= \{ g \in D(A^*_0) \colon A^*_0\, g \in X^\odot\}, \quad A^\odot_0 \, g = A^*_0\, g, 
\end{equation*}
see Proposition~3.10 in Appendix~II. 
In this particular case, 
\begin{equation*}
    A^\odot_0 \, g = g', \quad D(A^\odot_0)= \{ g \in C^1([0,1],\mathbb{R}) \colon g(1)=0=g'(1) \}. 
\end{equation*}

We now repeat this procedure of taking first the adjoint, and next a restriction. We represent $X^{\odot*} = NBV_0([-1,0],\mathbb{R})$, where the normalization of an element $\Phi$ means that $\Phi(0)=0$, $\Phi$ is continuous from the right on $[-1,0)$, so $\Phi$ does \emph{not} have a jump in $-1$ (since such a jump is irrelevant because $g(1)=0$), with the pairing being given by 
\begin{equation*}
    \langle \Phi,g \rangle = \int_{-1}^0 g(-\theta) \Phi(\dd \theta). 
\end{equation*}

It follows that 
\begin{equation*}
    (T^{\odot*}_0(t)\Phi)(\theta) = 
    \begin{cases}
        \Phi(t+\theta) & \theta\leq -t,\\
        0 & \theta > -t,
    \end{cases}
\end{equation*}
and
\begin{align*}
    &A^{\odot*}_0\Phi=\Phi', \\
    &D(A^{\odot*}_0) = \{ \Phi \colon \exists \Phi' \in X^{\odot*} \text{ such that } \Phi(\theta) = -\int_\theta^0 \Phi'(\alpha)\dd\alpha \}.
\end{align*}

When taking the closure of $D(A^{\odot*}_0)$ in order to characterize $X^{\odot\odot}$, the condition $\Phi'\in X^{\odot*}$ is weakened to $\Phi'\in L^1$. So 
\begin{equation*}
    X^{\odot\odot} = \{ \Phi \colon \exists \Phi' \in L^1([-1,0],\mathbb{R}) \text{ such that } \Phi(\theta) = -\int_\theta^0 \Phi'(\alpha) \dd \alpha \},
\end{equation*}
and therefore $X^{\odot\odot} \simeq L^1([-1,0],\mathbb{R}) = X$. 

Quite in general, we can define an embedding 
\begin{equation*}
    j \colon X \to X^{\odot*}
\end{equation*}
via 
\begin{equation*}
    \langle jx,x^\odot \rangle := \langle x^\odot, x\rangle. 
\end{equation*}

\begin{definition}
    $X$ is called $\odot$-reflexive with respect to $\{T_0(t)\}_{t\geq 0}$ when $j(X)=X^{\odot\odot}$. 
\end{definition}

\begin{theorem}
    $X$ is $\odot$-reflexive with respect to $\{T_0(t)\}_{t\geq 0}$ if and only if $(\lambda I - A_0)^{-1}$ is weakly compact for $\lambda \in \rho(A_0)$. 
\end{theorem}

\begin{corollary}
For delay equations, $X$ is $\odot$-reflexive with respect to \\$\{T_0(t)\}_{t\geq 0}$ if and only if the delay is finite. 
\end{corollary}

\begin{exercise}
    Verify that, in the present setting, $j(\phi)(\theta) = -\int_\theta^0 \phi(\alpha)\dd \alpha$. 
\end{exercise}

Conclusion: by taking twice adjoints and a restriction, we have come full circle. 

The rule for extension applies in the boundary point $\theta=0$ of the interval $[-1,0]$. 
Within the setting of $X=L^1$, we cannot `talk' about point values. But in the setting of $X^{\odot*}=NBV$ we have the Dirac measure concentrated in $\theta=0$ at our disposal. 
It is represented by 
\begin{equation*}
    r^{\odot*}(\theta) = 
    \begin{cases}
        -1 & \theta<0 , \\
        0 & \theta=0.
    \end{cases}
\end{equation*}
Let the rule for extension be given by $F\colon X\to \mathbb{R}$ (or, in other words, consider the RE~\eqref{RE-nonlinear}). 
Below we show that the initial value problem 
\begin{equation}\tag{RE} \label{RE-IVP}
    \begin{cases}
        b(t)=F(b_t), & t>0 \\
        b(\theta) = \phi(\theta), & -1\leq\theta\leq 0, 
    \end{cases}
\end{equation}
and the \emph{abstract integral equation} (AIE)
\begin{equation}\tag{AIE} \label{AIE}
    u(t) = T_0(t)\phi + j^{-1} \int_0^t T^{\odot*}_0 (t-\tau) r^{\odot*} F(u(\tau)) \dd\tau
\end{equation}
are equivalent in the sense that, when $b$ is a solution of \eqref{RE-IVP}, then $u(t):=b_t$ satisfies \eqref{AIE}, while if $u$ is a solution of \eqref{AIE}, then 
\begin{equation*}
    b(t)=\begin{cases}
        \phi(t), & t \in [-1,0],\\
        u(t)(0), & t>0, 
    \end{cases}
\end{equation*}
satisfies~\eqref{RE-IVP}. 

But first, we explain some of the theoretical background. 

\begin{lemma}[Lemma III.2.3]
    Let $h\colon \mathbb{R}\to X^{\odot*}$ be norm continuous. Define, for $t\geq 0$, $v(t) = \int_0^t T^{\odot*}(t-\sigma)h(\sigma)\dd\sigma \in X^{\odot*}$ by 
    \begin{equation*}
        \langle v(t), x^\odot \rangle = \int_0^t \langle h(\sigma), T^\odot (t-\sigma)x^\odot \rangle \dd\sigma. 
    \end{equation*}
Then:
\begin{enumerate}
    \item $v(t) \in X^{\odot\odot}$ (smoothing);
    \item $t\mapsto v(t)$ is norm continuous; 
    \item if $X$ is $\odot$-reflexive with respect to $\{T(t)\}$ and $M$, $\omega$ are such that $\|T(t)\|\leq M \ee^{\omega t}$, then 
    \begin{equation*}
        \|j^{-1}v(t)\|_X \leq M \frac{\ee^{\omega t}-1}{\omega} \sup_{0\leq \sigma\leq t} \|h(\sigma)\|_{X^{\odot*}},
    \end{equation*}
    \item $\frac{1}{t} v(t) \xrightarrow{*} h(0)$ for $t\to 0^+$, where $\xrightarrow{*}$ indicates weak* convergence. 
\end{enumerate} 

\end{lemma}

\begin{theorem}[Theorem III.2.4]
    Let $X$ be $\odot$-reflexive with respect to $\{T_0(t)\}_{t\geq 0}$. Let $B\colon X \to X^{\odot*}$ be a bounded linear operator. The linear~AIE
    \begin{equation*}
    u(t) = T_0(t)x + j^{-1} \int_0^t T^{\odot*}_0 (t-\sigma) B u(\sigma) \dd\sigma
\end{equation*}
has a unique continuous solution $u=u(t,x)$ and, if we define $T(t)x = u(t,x)$, then:
\begin{enumerate}
    \item $\{T(t)\}_{t\geq 0}$ is a strongly continuous semigroup on $X$; 
    \item the $X^\odot$ space does not change; 
    \item $X$ is $\odot$-reflexive with respect to $\{T(t)\}_{t\geq 0}$; 
    \item\label{generator} Let $A$ denote the generator of $\{T(t)\}_{t\geq 0}$, then $D(A^{\odot*})=D(A^{\odot*}_0)$ and $A^{\odot*} = A^{\odot*}_0+Bj^{-1}$;
    \item\label{domain} $D(A)=\{ x \in j^{-1} D(A^{\odot*}_0) \colon A^{\odot*}_0 jx + Bx \in j(X)=X^{\odot\odot}\}$, and $Ax = j^{-1}(A^{\odot*}_0jx +Bx)$. 
\end{enumerate}

\end{theorem}

\begin{remark}
    When going from \ref{generator} to \ref{domain}, information about the action ends up in information about the domain. 
\end{remark}

\begin{theorem}[Section VII.5]
    Let $G \colon X \to X^{\odot*}$ be Fr\'echet differentiable. Define $\Sigma(t)x=u(t,x)$, where $u$ is now the unique solution of the \emph{nonlinear}~AIE 
    \begin{equation*}
    u(t) = T_0(t)x + j^{-1} \int_0^t T^{\odot*}_0 (t-\sigma) G(u(\sigma)) \dd\sigma.
    \end{equation*}
Assume that $\Sigma(t)\overline{x}=\overline{x}$, for all $t \geq 0$. Then, $\Sigma(t)$ is Fr\'echet differentiable in $\overline{x}$ and $D\Sigma(t)(\overline{x}) = T(t)$, where $T(t)$ is the solution map corresponding to the linear AIE with $B=DG(\overline{x})$. 
In words: `solving AIE' and `linearization' commute. 

\end{theorem}

As a result, the AIE framework allows us to prove local stability and bifurcation results (such as the principle of linearized stability, stable-, unstable-, and center-manifold theorems, and the Hopf bifurcation theorem) by copying ODE proofs. For extension to the non-$\odot$-reflexive case, see \citet{Diekmann2012Blending} and \citet{Janssens2025}. 

So, once we prove the equivalence of \eqref{RE-IVP} and \eqref{AIE}, all these stability and bifurcation results become available for \eqref{RE-IVP}. 
The following lemma is the key step for proving the equivalence. 

\begin{lemma}
    Let $f \colon \mathbb{R}_+\to\mathbb{R}$ be continuous (in fact $f \in L^1$ suffices). Then 
    \begin{equation*}
        \left( j^{-1} \left(\int_0^t T^{\odot*}_0(t-\tau) r^{\odot*}f(\tau)\dd\tau \right) \right)(\theta) = 
        \begin{cases}
            f(t+\theta), & -t \leq \theta \leq 0, \\
            0, & \theta < -t.
        \end{cases}
    \end{equation*}
\end{lemma}
\begin{proof}
From 
\begin{equation*}
    (T^\odot_0(t-\tau)g)(\sigma) = 
    \begin{cases}
        g(\sigma+t-\tau), & \sigma+t-\tau \leq 1, \\
        0, & \sigma+t-\tau > 1,
    \end{cases}
\end{equation*}
it follows that 
\begin{equation*}
    \langle r^{\odot*}, T^\odot_0(t-\tau)g \rangle = 
    \begin{cases}
        g(t-\tau), & t-\tau \leq 1, \\
        0, & t-\tau > 1. 
    \end{cases}
\end{equation*}
Hence, 
\begin{align*}
\int_0^t \langle r^{\odot*}, T^\odot_0(t-\tau)g \rangle f(\tau) \dd\tau 
    &= \int_{\max\{0,t-1\}}^t g(t-\tau) f(\tau) \dd\tau \\
    &= \int_0^{t-\max\{0,t-1\}} g(\sigma) f(t-\sigma) \dd\sigma = \langle g,\phi\rangle, 
\end{align*}
where, for $t\leq 1$, 
\begin{equation*}
    \phi(\theta) = \begin{cases}
        f(t+\theta), & -t\leq \theta \leq 0, \\
        0, & \theta < -t,
    \end{cases}
\end{equation*}
and, for $t>1$, $\phi(\theta) = f(t+\theta)$, $-1\leq \theta\leq 0$. 
\end{proof}

\begin{exercise}
    Complete the proof of \eqref{RE-IVP} $\Leftrightarrow $ \eqref{AIE}. 
\end{exercise}


\CCLsection{State-at-Birth and the Next Generation Operator}
\label{s:NGO}

Due to the recent COVID-19 pandemic, the notion of `reproduction number' has been upgraded from obscure mathematical jargon, used in the context of epidemiological and ecological models, to a general concept. The underlying idea is the following. 

Already at birth, individuals may differ from one another. (Here we should interpret `birth' \emph{in sensu lato}: when dealing with infectious diseases, we consider a host individual that gets infected as `being born'.) So, we characterize newborn individuals by their $i$-state-at-birth. This can capture, for instance, various forms of host heterogeneity.

Let the prevailing constant environmental conditions be given/known. Consider a newborn individual. Based on model assumptions, one can compute the expected lifetime number of her daughters and the distribution of the $i$-state-at-birth of these daughters. This information provides the kernel for an operator that describes how a generation begets the next generation. Appropriately this operator is called the NGO, for \emph{next generation operator}. The basic reproduction number $R_0$ is, by definition, the spectral radius of the NGO \citep{DiekmannHeesterbeek1990}.

Often it is not immediately clear how to relate the above, interpretation guided, description in words to a population dynamic model in terms of semigroups of operators and infinitesimal generators. The aim of this appendix is to clarify the connection by way of the simplest example, viz., the situation where, as for age-dependent population dynamics, there is but one $i$-state-at-birth.
     
Abstractly we deal, on the `big' space, with linear problems of the form 
\begin{equation*}
    \frac{\dd u}{\dd t} = Au,
\end{equation*}
where $A = A_0 + B$ and $B$ has one-dimensional range. 
Here, $A_0$ generates the semigroup describing development and survival, while $B$ describes reproduction.

When $q^*$ and $q$ are such that 
\begin{equation*}
    B \phi = \langle q^*,\phi\rangle \, q,
\end{equation*}
then $q$ describes the (fixed, i.e., independent of the $i$-state of the mother) distribution of the $i$-state-at-birth, and $q^*$ describes the $i$-state specific fertility. 
In the case of age structure, $q$ is the Dirac delta in zero (represented by a Heaviside function). 

Assume that the half-plane $\{\Re \lambda > - \delta\}$ belongs to the resolvent set of $A_0$ for some $\delta > 0$. 
The spectral problem $A \phi = \lambda \phi$ can be written in the form
\begin{equation*}
    (\lambda I - A_0) \phi = \langle q^*,\phi\rangle \,q.
\end{equation*}    
It follows that $\phi$ is a multiple, say $c$, of $(\lambda I - A_0)^{-1} q$, and next that 
\begin{equation*}
    c = c \, \langle q^*, (\lambda I - A_0)^{-1} q\rangle. 
\end{equation*} 
Since we want a nontrivial $\phi$, we must have $c \ne 0$, and therefore $\lambda$ has to be a root of the characteristic equation
\begin{equation*}
    1 = \langle q^*, (\lambda I - A_0)^{-1} q\rangle . 
\end{equation*}

Another aspect of the kind of problems we study is that the operator $(\lambda I - A_0)^{-1}$ is a positive operator, depending on $\lambda$, as a real variable, in a monotone decreasing way. So the characteristic equation has a real root $\gr > 0$ if and only if $R_0 > 1$, where
\begin{equation}\label{R0-q}
    R_0 := - \langle q^*,  A_0^{-1} q \rangle .
\end{equation}

To connect with the `lifetime reproduction' interpretation, one should observe that 
\begin{equation}\label{invA0}
-A_0^{-1} = \int_0^\infty T_0(a) \dd a,
\end{equation}
and recall that $T_0(a)$ describes the development between birth and age $a$, and that $q$ describes the (distribution of the) $i$-state-at-birth.

As a concrete example, recall Exercise~\ref{ex:PDE-integrated}. In that situation, we have
\begin{equation*}
    (T_0(t) \Psi)(a) =
    \begin{cases}
        0, & a<t, \\
        \int_0^{a-t} \frac{\F(\sigma+t)}{\F(\sigma)} \Psi(\dd\sigma), & a\geq t.  
    \end{cases}
\end{equation*}
The choice $\Psi = q = \tilde{H}$ yields 
\begin{equation*}
    (T_0(t) \tilde{H})(a) =
    \begin{cases}
        0, & a<t, \\
        \F(t), & a\geq t, 
    \end{cases}
\end{equation*}
and when we pair this with $q^* = \beta$ we obtain
\begin{equation*}
\int_0^\infty \beta(\sigma) (T_0(t) \tilde{H})(\dd \sigma) = \beta(t) \F(t).
\end{equation*}
Combination of this identity with \eqref{R0-q} and \eqref{invA0} yields the formula \eqref{R0-integral} for~$R_0$.

\medskip
We conclude with three remarks.
\begin{enumerate}[label=(\roman*),noitemsep,nolistsep]
\item For numerical purposes, we need to be able to compute:
    \begin{itemize}
        \item $(\lambda I - A_0)^{-1} q$ (to determine $R_0$, we only need to do this for $\lambda = 0$);
        \item the pairing with $q^*$ (i.e., a certain integral). 
    \end{itemize}
This can be done for instance by pseudo-spectral methods, as introduced in Section~\ref{s:pseudospectral}. We refer to \citet{Breda2020collocation, Breda2021efficient, Breda2022bivariate} for approaches to compute $R_0$ as the spectral radius of a finite-dimensional matrix that approximates the NGO, and to \citet{DeReggi2024,DeReggi2025} for an approach to compute $R_0$ as the spectral radius of a finite-dimensional matrix that approximates the NGO, which also applies when $B$ has infinite-dimensional range. 

\item One only needs a bit of linear algebra to deal with the situation where $B$ has finite-dimensional range, for instance for systems of equations. Assume that, for $j= 1,2,...,m$, elements $q_j$ of the $p$-state space and elements $q^*_j$ of its dual space exist such that
\begin{equation*}
    B \phi = \sum_{j=1}^m \langle q^*_j , \phi \rangle q_j.
\end{equation*}
Define the entries of the $m\times m$ next generation matrix (NGM) $L$ by
\begin{equation*}
    L_{ij} := \langle q^*_i, -A_0^{-1} q_j\rangle. 
\end{equation*}
These specify the expected number of daughters with $i$-state-at-birth~$i$ of an individual who herself had $i$-state-at-birth~$j$. The basic reproduction number $R_0$ is the spectral radius of the NGM $L$ and, since $L$ is a non-negative matrix, $R_0$ is in fact the dominant eigenvalue of $L$ (where `dominant' means that the modulus of the other eigenvalues does not exceed~$R_0$). 

\item For NGO with infinite-dimensional range in an epidemic context, see for instance \cite{Bootsma2018} and \cite{Diekmann2025}. In that context, one can either work with generations in terms of absolute numbers, or in terms of trait specific fractions of the host population. Mathematically this is equivalent, but the fraction-version shows up in the equation for the expected final size of the epidemic outbreak (and one can understand the connection in terms of the probabilistic interpretation of the final size equation). We refer to \citet{Djidjou2025remarks} for a unifying framework covering structured population models formulated in terms of hyperbolic, parabolic, and delay differential equations.

\end{enumerate}

\bibliographystyle{plainnat}
\bibliography{age-structure.bib}

\end{document}

%% file: Figures/fig-IC.tex
  \begin{tikzpicture}
    \draw [->] (-1,0) -- (3,0) node [anchor=north east] {age ($a$)};
    \draw [->] (0,-1.5) -- (0,2) node [anchor=east] {time ($t$)};
    \path [draw=black, ultra thick, text=black, densely dashed, ->] (0,-1) node [anchor=east] {$\phi$} -- (0.5,-0.5) node [anchor=south] {} ; 
    \path [draw=black, ultra thick, text=black, densely dashed] (0.5,-0.5) node [anchor=east] {} -- (1,0) node [anchor=south] {$\psi$} ; 
    \end{tikzpicture}%
\qquad\qquad 
    \begin{tikzpicture}
    \draw [->] (-1,0) -- (3,0) node [anchor=north east] {age ($a$)};
    \draw [->] (0,-1.5) -- (0,2) node [anchor=east] {time ($t$)};

    \path [draw=black, text=black, densely dashed,->] (1,0) node [anchor=north] {$\psi$} -- (2,1) node [anchor=west] {$n(t,a)$} ; 
    \path [draw=black, text=black, densely dashed,->] (0,1) node [anchor=east] {$b$} -- (1,2) node [anchor=east] {} ;     
    \path [draw=black, text=black, densely dashed,->] (0,0) node [anchor=east] {} -- (2,2) node [anchor=east] {} ; 
    
    \end{tikzpicture}

%% file: Figures/fig-lotka-euler.tex
  \begin{tikzpicture}
    \draw [->] (-1,0) -- (5,0) node [anchor=north east] {$\lambda$ (real)};
    \draw [->] (0,-1) -- (0,4) node [anchor=east] {};
    \draw[color=black,thick,domain=-1:5] plot (\x,{2*exp(-0.5*\x)}) ;
    \node [anchor = south] at (-0.5,3.2) {$\overline{k}(\lambda)$};
    \draw [-,color=gray,densely dashed] (-1,1) -- (4,1);
    \node [anchor = south east] at (0,1) {$1$};
    \filldraw[black] (0,2) circle (2pt) node[anchor=south west]{$R_0$};
    \draw [-,color=gray,densely dashed] (1.38,0) node [anchor=north,color=black] {$\gr$} -- (1.38,1);
    \end{tikzpicture}

%% file: Figures/fig-interpolation.tex

%
%
\definecolor{mycolor1}{rgb}{0.00000,0.44700,0.74100}%
\begin{tikzpicture}

\small 

\begin{axis}[%
axis y line=right,
axis x line* = left,
width=3.5in,
height=2in,
at={(0,0)},
xmin=-1,
xmax=0,
ymin=-5.66213785987333e-15,
ymax=2,
xtick={0,
-0.05,
-0.25,
-0.5,
-0.75,
-0.95,
-1},
xticklabels={$0$,$\theta_1$,$\theta_2$, , $\theta_{M-1}$,\hspace{15pt} $\theta_M$,{\hspace{-20pt} $-a_{max}$}},
ytick=\empty,
xmajorgrids,
axis background/.style={fill=white}
]
\addplot [color=black, only marks, mark=*, mark options={solid, black}, forget plot]
  table[row sep=crcr]{%
0	0\\
-0.05	0\\
-0.25	0\\
-0.5	0\\
-0.75	0\\
-0.95	0\\
};

\node[anchor = north east, align=left, inner sep=0mm, text=black]
at (axis cs:-0.05,0.52,0) {$y_1(t)$};
\node[anchor = south east, align=left, inner sep=0mm, text=black]
at (axis cs:-0.25,0.7,0) {$y_2(t)$};
\node[anchor = south west, align=left, inner sep=0mm, text=black]
at (axis cs:-0.75,1.15,0) {$y_{M-1}(t)$};
\node[anchor = south west, align=left, inner sep=0mm, text=black]
at (axis cs:-0.95,1.4,0) {$y_M(t)$};

\node[anchor = south, align=left, inner sep=0mm, text=black]
at (axis cs:-0.15,1.5,0) {$p$};

\draw [->,thick,color=black] (axis cs:-0.15,1.45,0) -- (axis cs:-0.12,0.8);

\addplot [color=black, only marks, mark=*, mark options={solid, black}, forget plot]
  table[row sep=crcr]{%
0	0\\
-0.05	0.525635548188012\\
-0.25	0.642012708343871\\
-0.5	0.824360635350064\\
-0.75	1.05850000830634\\
-0.95	1.29285482965792\\
};

\addplot [color=black, line width=1.0pt, forget plot]
  table[row sep=crcr]{%
-1	1.85913576856716\\
-0.99	1.71300745237959\\
-0.98	1.58469062665893\\
-0.97	1.47277765704871\\
-0.96	1.37592559486629\\
-0.95	1.29285482965792\\
-0.94	1.22234774175379\\
-0.93	1.16324735482309\\
-0.92	1.11445598842909\\
-0.91	1.07493391058416\\
-0.9	1.04369799030489\\
-0.89	1.01982035016708\\
-0.88	1.00242701886085\\
-0.87	0.9906965837457\\
-0.86	0.98385884340553\\
-0.85	0.981193460203739\\
-0.84	0.982028612838271\\
-0.83	0.985739648896672\\
-0.82	0.991747737411156\\
-0.81	0.999518521413663\\
-0.8	1.00856077049092\\
-0.79	1.01842503333949\\
-0.78	1.02870229032085\\
-0.77	1.03902260601646\\
-0.76	1.04905378178277\\
-0.75	1.05850000830634\\
-0.74	1.06710051815887\\
-0.73	1.07462823835227\\
-0.72	1.08088844289371\\
-0.71	1.0857174053407\\
-0.7	1.08898105135611\\
-0.69	1.09057361126329\\
-0.68	1.09041627260107\\
-0.67	1.08845583267885\\
-0.66	1.08466335113168\\
-0.65	1.07903280247527\\
-0.64	1.07157972866108\\
-0.63	1.0623398916314\\
-0.62	1.05136792587436\\
-0.61	1.03873599097903\\
-0.6	1.02453242419046\\
-0.59	1.00886039296476\\
-0.58	0.991836547524141\\
-0.57	0.973589673411972\\
-0.56	0.954259344047862\\
-0.55	0.933994573282701\\
-0.54	0.91295246795373\\
-0.53	0.891296880439597\\
-0.52	0.869197061215415\\
-0.51	0.846826311407827\\
-0.5	0.824360635350063\\
-0.49	0.801977393137005\\
-0.48	0.779853953180236\\
-0.47	0.758166344763112\\
-0.46	0.737087910595815\\
-0.45	0.716787959370417\\
-0.44	0.697430418315937\\
-0.43	0.679172485753402\\
-0.42	0.662163283650908\\
-0.41	0.64654251017868\\
-0.4	0.632439092264132\\
-0.39	0.619969838146925\\
-0.38	0.609238089934029\\
-0.37	0.600332376154784\\
-0.36	0.593325064315958\\
-0.35	0.588271013456808\\
-0.34	0.585206226704139\\
-0.33	0.584146503827368\\
-0.32	0.585086093793578\\
-0.31	0.587996347322583\\
-0.3	0.592824369441985\\
-0.29	0.599491672042237\\
-0.28	0.6078928264317\\
-0.27	0.617894115891705\\
-0.26	0.629332188231611\\
-0.25	0.642012708343871\\
-0.24	0.655709010759082\\
-0.23	0.670160752201055\\
-0.22	0.68507256414187\\
-0.21	0.700112705356934\\
-0.2	0.714911714480049\\
-0.19	0.729061062558462\\
-0.18	0.742111805607934\\
-0.17	0.753573237167793\\
-0.16	0.762911540856001\\
-0.15	0.769548442924206\\
-0.14	0.77285986481281\\
-0.13	0.772174575706023\\
-0.12	0.766772845086927\\
-0.11	0.755885095292533\\
-0.1	0.738690554068846\\
-0.09	0.714315907125919\\
-0.08	0.681833950692915\\
-0.07	0.640262244073169\\
-0.06	0.58856176219925\\
-0.05	0.525635548188012\\
-0.04	0.450327365895665\\
-0.03	0.361420352472828\\
-0.02	0.257635670919592\\
-0.01	0.137631162640578\\
0	-3.26249847394851e-15\\
};

\end{axis}
\end{tikzpicture}%


%% file: Figures/fig-error.tex
%
%
%
\definecolor{mycolor1}{rgb}{0.00000,0.44700,0.74100}%
\begin{tikzpicture}

\begin{axis}[%
axis y line*=left,
axis x line*=left,
width=3.5in,
height=2.5in,
at={(0,0)},
xmode=log,
xmin=1,
xmax=50,
xminorticks=true,
ymode=log,
ymin=1e-10,
ymax=100,
yminorticks=true,
axis background/.style={fill=white},
xtick={1,2,3,4,5,6,7,8,9,10,20,30,40,50},
xticklabels={1,,,,,,,,,10,,,,50},
xticklabel style={font=\small},
yticklabel style={font=\small},
xlabel style={font=\small},
xlabel=$M$,
grid=both,
]
\addplot [color=black, mark=*, mark options={solid, black}, forget plot]
  table[row sep=crcr]{%
1	nan\\
2	2.94129858516234\\
3	1.32965979552885\\
4	5.95592905006689\\
5	9.12092461410465\\
6	12.4235120806191\\
7	9.20866304753386\\
8	5.87727808969711\\
9	0.477693389067397\\
10	0.412772507520373\\
11	3.51427425169077\\
12	0.0184154997472277\\
13	0.00250361032962587\\
14	0.000655821965658276\\
15	8.28773437575592e-05\\
16	1.94857168835938e-05\\
17	5.41769635020728e-06\\
18	2.6226560834175e-07\\
19	7.03046985961464e-07\\
20	1.11131584645818e-06\\
21	7.13729662038531e-08\\
22	5.52236826933949e-07\\
23	8.11235430830948e-07\\
24	2.3236331969656e-07\\
25	3.21453732965438e-07\\
26	2.27929270124605e-07\\
27	4.54158382012793e-08\\
28	3.45708379256848e-07\\
29	3.42921765650317e-07\\
30	2.79401596259277e-06\\
31	3.12903395638386e-07\\
32	1.61802020670621e-07\\
33	4.47140280357416e-09\\
34	2.96679453981596e-06\\
35	3.04563563702231e-07\\
36	3.06031580521449e-07\\
37	4.06077163006557e-07\\
38	1.8551720160076e-07\\
39	4.50988650868567e-08\\
40	2.93550563412737e-07\\
41	2.15205595566204e-07\\
42	3.87494985432113e-07\\
43	1.8667446965992e-07\\
44	2.50550051816845e-07\\
45	1.27715973974318e-06\\
46	3.2015780391248e-06\\
47	1.39704567914123e-07\\
48	1.96387738071735e-07\\
49	1.70639409446949e-06\\
50	0\\
};
\end{axis}
\end{tikzpicture}%


%% file: main.bbl
\begin{thebibliography}{56}
\providecommand{\natexlab}[1]{#1}
\providecommand{\url}[1]{\texttt{#1}}
\expandafter\ifx\csname urlstyle\endcsname\relax
  \providecommand{\doi}[1]{doi: #1}\else
  \providecommand{\doi}{doi: \begingroup \urlstyle{rm}\Url}\fi

\bibitem[And{\`o} et~al.(2022)And{\`o}, Breda, Liessi, Maset, Scarabel, and
  Vermiglio]{Ando202215years}
A.~And{\`o}, D.~Breda, D.~Liessi, S.~Maset, F.~Scarabel, and R.~Vermiglio.
\newblock 15 years or so of pseudospectral collocation methods for stability
  and bifurcation of delay equations.
\newblock In \emph{Accounting for Constraints in Delay Systems}, pages
  127--149. Springer Cham, 2022.

\bibitem[And{\`o} et~al.(2023)And{\`o}, {De Reggi}, Liessi, and
  Scarabel]{Ando2022pseudospectral}
A.~And{\`o}, S.~{De Reggi}, D.~Liessi, and F.~Scarabel.
\newblock A pseudospectral method for investigating the stability of linear
  population models with two physiological structures.
\newblock \emph{Mathemtical Biosciences and Engineering}, 20:\penalty0
  4493--4515, 2023.

\bibitem[Arendt et~al.(1986)Arendt, Grabosch, Greiner, Moustakas, Nagel,
  Schlotterbeck, Groh, Lotz, and Neubrander]{Arendt1986}
W.~Arendt, A.~Grabosch, G.~Greiner, U.~Moustakas, R.~Nagel, U.~Schlotterbeck,
  U.~Groh, H.~P. Lotz, and F.~Neubrander.
\newblock \emph{One-Parameter Semigroups of Positive Operators}, volume 1184 of
  \emph{Lecture Notes in Mathematics}.
\newblock Springer Berlin, Heidelberg, 1986.

\bibitem[Barril et~al.(2022)Barril, Calsina, Diekmann, and Farkas]{Barril2022}
C.~Barril, {\`A}.~Calsina, O.~Diekmann, and J.~Z. Farkas.
\newblock On the formulation of size-structured consumer resource models (with
  special attention for the principle of linearized stability).
\newblock \emph{Mathematical Models and Methods in Applied Sciences},
  32\penalty0 (6):\penalty0 1141--1191, 2022.

\bibitem[B{\'a}tkai et~al.(2017)B{\'a}tkai, Fijav{\v{z}}, and
  Rhandi]{Batkai2017}
A.~B{\'a}tkai, M.~K. Fijav{\v{z}}, and A.~Rhandi.
\newblock \emph{Positive Operator Semigroups}, volume 257 of \emph{Operator
  Theory: Advances and Applications}.
\newblock Birkh\"auser Cham, 2017.

\bibitem[Bellman and Cooke(1963)]{BellmanCooke1963}
R.~Bellman and K.~L. Cooke.
\newblock Differential--difference equations.
\newblock \emph{The RAND Corporation}, 1963.

\bibitem[Bootsma et~al.(2024)Bootsma, Chan, Diekmann, and Inaba]{Bootsma2018}
M.~C.~J. Bootsma, K.~M.~D. Chan, O.~Diekmann, and H.~Inaba.
\newblock The effect of host population heterogeneity on epidemic outbreaks.
\newblock \emph{Mathematics in Applied Sciences and Engineering}, 5:\penalty0
  1--35, 2024.

\bibitem[Breda et~al.(2015)Breda, Maset, and Vermiglio]{BMVbook}
D.~Breda, S.~Maset, and R.~Vermiglio.
\newblock \emph{Stability of Linear Delay Differential Equations: A Numerical
  Approach with MATLAB}.
\newblock SpringerBriefs in Control, Automation and Robotics. Springer New
  York, NY, 2015.
\newblock ISBN 9781493921072.

\bibitem[Breda et~al.(2016)Breda, Diekmann, Gyllenberg, Scarabel, and
  Vermiglio]{Breda2016Prospects}
D.~Breda, O.~Diekmann, M.~Gyllenberg, F.~Scarabel, and R.~Vermiglio.
\newblock Pseudospectral discretization of nonlinear delay equations: new
  prospects for numerical bifurcation analysis.
\newblock \emph{SIAM Journal on Applied Dynamical Systems}, 15\penalty0
  (1):\penalty0 1--23, 2016.

\bibitem[Breda et~al.(2020)Breda, Kuniya, Ripoll, and
  Vermiglio]{Breda2020collocation}
D.~Breda, T.~Kuniya, J.~Ripoll, and R.~Vermiglio.
\newblock Collocation of next-generation operators for computing the basic
  reproduction number of structured populations.
\newblock \emph{Journal of Scientific Computing}, 85\penalty0 (2):\penalty0 40,
  2020.

\bibitem[Breda et~al.(2021)Breda, Florian, Ripoll, and
  Vermiglio]{Breda2021efficient}
D.~Breda, F.~Florian, J.~Ripoll, and R.~Vermiglio.
\newblock Efficient numerical computation of the basic reproduction number for
  structured populations.
\newblock \emph{Journal of Computational and Applied Mathematics},
  384:\penalty0 113165, 2021.

\bibitem[Breda et~al.(2022)Breda, {De Reggi}, Scarabel, Vermiglio, and
  Wu]{Breda2022bivariate}
D.~Breda, S.~{De Reggi}, F.~Scarabel, R.~Vermiglio, and J.~Wu.
\newblock Bivariate collocation for computing {$R_0$} in epidemic models with
  two structures.
\newblock \emph{Computers \& Mathematics with Applications}, 116:\penalty0
  15--24, 2022.

\bibitem[Butzer and Berens(1967)]{ButzerBerens1967}
P.~P. Butzer and H.~Berens.
\newblock \emph{Semi-Groups of Operators and Approximation}, volume 145 of
  \emph{Grundlehren der Mathematischen Wissenschaften}.
\newblock Springer Berlin, Heidelberg, 1967.

\bibitem[Cassidy(2025)]{Cassidy2025using}
T.~Cassidy.
\newblock Using multidelay discrete delay differential equations to accurately
  simulate models with distributed delays.
\newblock \emph{Studies in Applied Mathematics}, 154\penalty0 (6):\penalty0
  e70069, 2025.

\bibitem[Cassidy et~al.(2022)Cassidy, Gillich, Humphries, and van
  Dorp]{Cassidy2022numerical}
T.~Cassidy, P.~Gillich, A.~R. Humphries, and C.~H. van Dorp.
\newblock Numerical methods and hypoexponential approximations for {Gamma}
  distributed delay differential equations.
\newblock \emph{IMA Journal of Applied Mathematics}, 87\penalty0 (6):\penalty0
  1043--1089, 2022.

\bibitem[Davis(1975)]{Davis1975book}
D.~J. Davis.
\newblock \emph{Interpolation and approximation}.
\newblock Courier Corporation, 1975.

\bibitem[{De Reggi} et~al.(2024{\natexlab{a}}){De Reggi}, Scarabel, and
  Vermiglio]{DeReggi2024}
S.~{De Reggi}, F.~Scarabel, and R.~Vermiglio.
\newblock Approximating reproduction numbers: a general numerical approach for
  age-structured models.
\newblock \emph{Mathematical Biosciences and Engineering}, 21\penalty0
  (4):\penalty0 5360--5393, 2024{\natexlab{a}}.

\bibitem[{De Reggi} et~al.(2024{\natexlab{b}}){De Reggi}, Scarabel, and
  Vermiglio]{DeReggi2025}
S.~{De Reggi}, F.~Scarabel, and R.~Vermiglio.
\newblock On the convergence of the pseudospectral approximation of
  reproduction numbers for age-structured models.
\newblock \emph{arXiv}, arXiv:2409.01520, 2024{\natexlab{b}}.

\bibitem[De~Roos and Persson(2013)]{DeRoosPersson2013}
A.~M. De~Roos and L.~Persson.
\newblock \emph{Population and Community Ecology of Ontogenetic Development},
  volume~51 of \emph{Monographs in Population Biology}.
\newblock Princeton University Press, 2013.

\bibitem[{De Wolff} et~al.(2021){De Wolff}, Scarabel, Verduyn~Lunel, and
  Diekmann]{DeWolff2021}
B.~A.~J. {De Wolff}, F.~Scarabel, S.~M. Verduyn~Lunel, and O.~Diekmann.
\newblock Pseudospectral approximation of {Hopf} bifurcation for delay
  differential equations.
\newblock \emph{SIAM Journal on Applied Dynamical Systems}, 20\penalty0
  (1):\penalty0 333--370, 2021.

\bibitem[Dhooge et~al.(2008)Dhooge, Govaerts, Kuznetsov, Meijer, and
  Sautois]{MatCont2008}
A.~Dhooge, W.~Govaerts, Yu.~A. Kuznetsov, H.~G.~E. Meijer, and B.~Sautois.
\newblock New features of the software {MatCont} for bifurcation analysis of
  dynamical systems.
\newblock \emph{Mathematical and Computer Modelling of Dynamical Systems},
  14\penalty0 (2):\penalty0 147--175, 2008.

\bibitem[Diekmann and Gyllenberg(2012)]{Diekmann2012Blending}
O.~Diekmann and M.~Gyllenberg.
\newblock Equations with infinite delay: blending the abstract and the
  concrete.
\newblock \emph{Journal of Differential Equations}, 252\penalty0 (2):\penalty0
  819--851, 2012.

\bibitem[Diekmann and Korvasov\'a(2013)]{Diekmann2013didactical}
O.~Diekmann and K.~Korvasov\'a.
\newblock A didactical note on the advantage of using two parameters in {Hopf}
  bifurcation studies.
\newblock \emph{Journal of Biological Dynamics}, 7\penalty0 (sup1):\penalty0
  21--30, 2013.

\bibitem[Diekmann and Scarabel(2025)]{DiekmannScarabelSize}
O.~Diekmann and F.~Scarabel.
\newblock Size-structured population dynamics.
\newblock \emph{arXiv}, arXiv:2506.03413, 2025.

\bibitem[Diekmann and Verduyn~Lunel(2021)]{Diekmann2021Twin}
O.~Diekmann and S.~M. Verduyn~Lunel.
\newblock Twin semigroups and delay equations.
\newblock \emph{Journal of Differential Equations}, 286:\penalty0 332--410,
  2021.

\bibitem[Diekmann et~al.(1990)Diekmann, Heesterbeek, and
  Metz]{DiekmannHeesterbeek1990}
O.~Diekmann, J.~A.~P. Heesterbeek, and J.~A.~J. Metz.
\newblock On the definition and the computation of the basic reproduction ratio
  {$R_0$} in models for infectious diseases in heterogeneous populations.
\newblock \emph{Journal of Mathematical Biology}, 28:\penalty0 365--382, 1990.

\bibitem[Diekmann et~al.(1995)Diekmann, Van~Gils, Verduyn~Lunel, and
  Walther]{Diekmann1995Delay}
O.~Diekmann, S.~A. Van~Gils, S.~M. Verduyn~Lunel, and H.-O. Walther.
\newblock \emph{Delay Equations: Functional-, Complex-, and Nonlinear
  Analysis}, volume 110 of \emph{Applied Mathematical Sciences}.
\newblock Springer New York, NY, 1995.

\bibitem[Diekmann et~al.(2008)Diekmann, Getto, and
  Gyllenberg]{Diekmann2008Suns}
O.~Diekmann, P.~Getto, and M.~Gyllenberg.
\newblock Stability and bifurcation analysis of {Volterra} functional equations
  in the light of suns and stars.
\newblock \emph{SIAM Journal on Mathematical Analysis}, 39\penalty0
  (4):\penalty0 1023--1069, 2008.

\bibitem[Diekmann et~al.(2013)Diekmann, Heesterbeek, and
  Britton]{Diekmann2013EpiBook}
O.~Diekmann, H.~Heesterbeek, and T.~Britton.
\newblock \emph{Mathematical Tools for Understanding Infectious Disease
  Dynamics}.
\newblock Princeton Series in Theoretical and Computational Biology. Princeton
  University Press, 2013.

\bibitem[Diekmann et~al.(2018)Diekmann, Gyllenberg, and
  Metz]{Diekmann2018finite}
O.~Diekmann, M.~Gyllenberg, and J.~A.~J. Metz.
\newblock Finite dimensional state representation of linear and nonlinear delay
  systems.
\newblock \emph{Journal of Dynamics and Differential Equations}, 30:\penalty0
  1439--1467, 2018.

\bibitem[Diekmann et~al.(2025)Diekmann, Inaba, and Thieme]{Diekmann2025}
O.~Diekmann, H.~Inaba, and H.~R. Thieme.
\newblock Mathematical epidemiology of infectious diseases: an ongoing
  challenge.
\newblock \emph{Japan Journal of Industrial and Applied Mathematics}, 2025.
\newblock URL \url{\url{https://doi.org/10.1007/s13160-025-00742-1}}.

\bibitem[Djidjou-Demasse et~al.(2025)Djidjou-Demasse, Ducrot, Pane, and
  Seydi]{Djidjou2025remarks}
R.~Djidjou-Demasse, A.~Ducrot, M.~Pane, and O.~Seydi.
\newblock Remarks on invasion threshold for structured population models.
\newblock \emph{hal.science}, 2025.

\bibitem[Engel and Nagel(2000)]{EngelNagel2000}
K.-J. Engel and R.~Nagel.
\newblock \emph{One-Parameter Semigroups for Linear Evolution Equations},
  volume 194 of \emph{Graduate Texts in Mathematics}.
\newblock Springer New York, NY, 2000.

\bibitem[Feller(1941)]{Feller1941}
W.~Feller.
\newblock On the integral equation of renewal theory.
\newblock \emph{The Annals of Mathematical Statistics}, 12:\penalty0 243--267,
  1941.

\bibitem[Feller(1971)]{Feller1971volII}
W.~Feller.
\newblock \emph{An Introduction to Probability Theory and Its Applications,
  vol.~II}.
\newblock Wiley, New York, 1971.

\bibitem[Gautschi(2000)]{Gautschi2004book}
W.~Gautschi.
\newblock \emph{Orthogonal Polynomials: Computation and Approximation}.
\newblock Oxford Academic, 2000.

\bibitem[Gripenberg et~al.(1990)Gripenberg, Londen, and
  Staffans]{Gripenberg1990}
G.~Gripenberg, S.-O. Londen, and O.~Staffans.
\newblock \emph{Volterra Integral and Functional Equations}.
\newblock Cambridge University Press, 1990.

\bibitem[Guo and Wu(2013)]{GuoWu2013}
G.~Guo and J.~Wu.
\newblock \emph{Bifurcation Theory of Functional Differential Equations},
  volume 184 of \emph{Applied Mathematical Sciences}.
\newblock Springer New York, NY, 2013.

\bibitem[Gurney et~al.(1980)Gurney, Blythe, and Nisbet]{Gurney1980}
W.~S.~C. Gurney, S.~P. Blythe, and R.~M. Nisbet.
\newblock Nicholson's blowflies revisited.
\newblock \emph{Nature}, 287:\penalty0 17--21, 1980.

\bibitem[Gyllenberg et~al.(2018)Gyllenberg, Scarabel, and
  Vermiglio]{Gyllenberg2018}
M.~Gyllenberg, F.~Scarabel, and R.~Vermiglio.
\newblock Equations with infinite delay: Numerical bifurcation analysis via
  pseudospectral discretization.
\newblock \emph{Applied Mathematics and Computation}, 333:\penalty0 490--505,
  2018.

\bibitem[Inaba(2017)]{Inaba2017book}
H.~Inaba.
\newblock \emph{Age-Structured Population Dynamics in Demography and
  Epidemiology}.
\newblock Springer Singapore, 2017.

\bibitem[Janssens(2020)]{Janssens2025}
S.~G. Janssens.
\newblock A class of abstract delay differential equations in the light of suns
  and stars. ii.
\newblock \emph{arXiv}, arXiv:2003.13341, 2020.

\bibitem[Kermack and McKendrick(1927)]{KermackMcKendrick1927}
W.~O. Kermack and A.~G. McKendrick.
\newblock A contribution to the mathematical theory of epidemics.
\newblock \emph{Proceedings of the Royal Society of London. Series A,
  Containing papers of a mathematical and physical character}, 115\penalty0
  (772):\penalty0 700--721, 1927.

\bibitem[Lasota and Mackey(1994)]{LasotaMackey2013Book}
A.~Lasota and M.~C. Mackey.
\newblock \emph{Chaos, Fractals, and Noise: Stochastic Aspects of Dynamics},
  volume~97 of \emph{Applied Mathematical Sciences}.
\newblock Springer New York, NY, 1994.

\bibitem[Liessi et~al.(2025)Liessi, Santi, Vermiglio, Thakur, Meijer, and
  Scarabel]{Liessi-matcont}
D.~Liessi, E.~Santi, R.~Vermiglio, M.~Thakur, H.~G.~E. Meijer, and F.~Scarabel.
\newblock New functionalities in {MatCont}: delay equations and {Lyapunov}
  exponents.
\newblock \emph{arXiv}, arXiv:2504.12785, 2025.

\bibitem[Lotka(1939)]{Lotka1939}
A.~J. Lotka.
\newblock On an integral equation in population analysis.
\newblock \emph{The Annals of Mathematical Statistics}, 10\penalty0
  (2):\penalty0 144--161, 1939.

\bibitem[Magal and Ruan(2018)]{MagalRuan2018book}
P.~Magal and S.~Ruan.
\newblock \emph{Theory and Applications of Abstract Semilinear Cauchy
  Problems}, volume 201 of \emph{Applied Mathematical Sciences}.
\newblock Springer Cham, 2018.

\bibitem[Metz and Diekmann(1986)]{MetzDiekmann1986}
J.~A.~J. Metz and O.~Diekmann.
\newblock \emph{The Dynamics of Physiologically Structured Populations},
  volume~68 of \emph{Lecture Notes in Biomathematics}.
\newblock Springer Berlin, Heidelberg, 1986.

\bibitem[Perthame(2006)]{Perthame2006Book}
B.~Perthame.
\newblock \emph{Transport Equations in Biology}.
\newblock Birkh\"auser Basel, 2006.

\bibitem[Rudin(1973)]{Rudin1973Book}
W.~Rudin.
\newblock \emph{Functional Analysis}.
\newblock McGraw-Hill, 1973.

\bibitem[Rudnicki and Tyran-Kami{\'n}ska(2017)]{Rudnicki2017Book}
R.~Rudnicki and M.~Tyran-Kami{\'n}ska.
\newblock \emph{Piecewise Deterministic Processes in Biological Models}.
\newblock SpringerBriefs in Mathematical Methods. Springer Cham, 2017.

\bibitem[Scarabel and Vermiglio(2024)]{Scarabel2024Infinite}
F.~Scarabel and R.~Vermiglio.
\newblock Equations with infinite delay: pseudospectral discretization for
  numerical stability and bifurcation in an abstract framework.
\newblock \emph{SIAM Journal on Numerical Analysis}, 62\penalty0 (4):\penalty0
  1736--1758, 2024.

\bibitem[Scarabel et~al.(2021{\natexlab{a}})Scarabel, Breda, Diekmann,
  Gyllenberg, and Vermiglio]{Scarabel2021Vietnam}
F.~Scarabel, D.~Breda, O.~Diekmann, M.~Gyllenberg, and R.~Vermiglio.
\newblock Numerical bifurcation analysis of physiologically structured
  population models via pseudospectral approximation.
\newblock \emph{Vietnam Journal of Mathematics}, 49:\penalty0 37--67,
  2021{\natexlab{a}}.

\bibitem[Scarabel et~al.(2021{\natexlab{b}})Scarabel, Diekmann, and
  Vermiglio]{Scarabel2021JCAM}
F.~Scarabel, O.~Diekmann, and R.~Vermiglio.
\newblock Numerical bifurcation analysis of renewal equations via
  pseudospectral approximation.
\newblock \emph{Journal of Computational and Applied Mathematics},
  397:\penalty0 113611, 2021{\natexlab{b}}.

\bibitem[Thieme(1990)]{Thieme1990integrated}
H.~R. Thieme.
\newblock ``{Integrated} semigroups'' and integrated solutions to abstract
  {Cauchy} problems.
\newblock \emph{Journal of Mathematical Analysis and Applications},
  152\penalty0 (2):\penalty0 416--447, 1990.

\bibitem[Trefethen(2000)]{Trefethen2000spectral}
L.~N. Trefethen.
\newblock \emph{Spectral methods in {MATLAB}}.
\newblock Software, Environments, and Tools. SIAM, 2000.

\bibitem[Webb(1985)]{Webb1985Book}
G.~F. Webb.
\newblock \emph{Theory of Nonlinear Age-Dependent Population Dynamics}.
\newblock CRC Press, 1985.

\end{thebibliography}
